\documentclass[11pt]{article}

\usepackage[utf8]{inputenc}
\usepackage[margin=1in]{geometry}
\usepackage{amsmath}
\usepackage{amssymb}
\usepackage{amsthm}
\usepackage{mathtools}
\usepackage{bbm}
\usepackage{bm}
\usepackage{color}
\usepackage{graphicx}
\usepackage{array}
\usepackage{multirow}
\usepackage{enumitem}
\usepackage[ruled,vlined]{algorithm2e}
\usepackage{relsize}
\usepackage[usenames,dvipsnames]{xcolor}
\usepackage[linktocpage=true,
	pagebackref=true,colorlinks,
	linkcolor=BrickRed,citecolor=blue]
	{hyperref}
\usepackage{cleveref}
\usepackage[nottoc]{tocbibind}

\usepackage{thmtools,thm-restate}
\usepackage{nicefrac}
\usepackage{makecell}

\newcommand{\vsthm}{\vspace{-0.3em}}

\DeclareMathOperator*{\E}{\mathbb{E}}
\DeclareMathOperator*{\Var}{Var}

\DeclareMathOperator{\given}{\:\vert\:}
\DeclareMathOperator{\mutual}{\,;\,}

\DeclareMathOperator{\OPT}{OPT}
\DeclareMathOperator{\ALG}{ALG}
\DeclareMathOperator{\MIR}{MIR}
\DeclareMathOperator{\MECH}{MECH}

\DeclareMathOperator{\poly}{poly}
\DeclareMathOperator{\Binom}{Binom}
\newcommand{\MIRsubmod}{\ensuremath{\mathsf{MIR_{SubMod}}}}
\newcommand{\MIRgen}{\ensuremath{\mathsf{MIR_{Gen}}}}

\newcommand{\calA}{\mathcal{A}}
\newcommand{\calB}{\mathcal{B}}

\newcommand{\calF}{\mathcal{F}}

\newcommand{\calM}{\mathcal{M}}
\newcommand{\calN}{\mathcal{N}}

\newcommand{\calP}{\mathcal{P}}

\newcommand{\calS}{\mathcal{S}}

\newcommand{\calX}{\mathcal{X}}

\let\abs\relax
\DeclarePairedDelimiter{\abs}{\lvert}{\rvert}

\newcommand{\IGNORE}[1]{}

\newcounter{note}[section]

\newcommand{\PreserveBackslash}[1]{\let\temp=\\#1\let\\=\temp}
\newcolumntype{C}[1]{>{\PreserveBackslash\centering}p{#1}}
\newcolumntype{R}[1]{>{\PreserveBackslash\raggedleft}p{#1}}
\newcolumntype{L}[1]{>{\PreserveBackslash\raggedright}p{#1}}

\newtheorem{theorem}{Theorem}[section]
\newtheorem{claim}[theorem]{Claim}
\Crefname{claim}{Claim}{Claims}
\newtheorem{proposition}[theorem]{Proposition}
\newtheorem{lemma}[theorem]{Lemma}

\theoremstyle{definition}
\newtheorem{example}[theorem]{Example}

\newtheorem{definition}[theorem]{Definition}

\newtheorem{remark}[theorem]{Remark}

\title{Settling the Communication Complexity of VCG-based Mechanisms for all Approximation Guarantees}
\author{Frederick V. Qiu \and S.~Matthew Weinberg}
\date{\today}

\begin{document}

\maketitle

\begin{abstract}

We consider truthful combinatorial auctions with items $M \coloneqq [m]$ for sale to $n$ bidders, where each bidder $i$ has a private monotone valuation function $v_i: 2^M \to \mathbb{R}_+$. Among truthful mechanisms, \emph{maximal-in-range} (MIR) mechanisms (sometimes called \emph{VCG-based}) achieve the best-known approximation guarantees among all poly-communication deterministic truthful mechanisms in all previously-studied settings. Our work settles the communication complexity necessary to achieve any approximation guarantee via an MIR mechanism. Specifically:

\medskip
\noindent Let $\MIRsubmod(m, k)$ denote the best approximation guarantee achievable by an MIR mechanism using $2^k$ communication between bidders with submodular valuations over $m$ items. Then
\begin{itemize}
    \item For all $k = \Omega(\log(m))$, $\MIRsubmod(m,k) = \Omega(\sqrt{m/(k\log(m/k))})$. When $k = \Theta(\log(m))$, this improves the previous best lower bound for polynomial communication maximal-in-range mechanisms from $\Omega(m^{1/3}/\log^{2/3}(m))$~\cite{DanielySS15} to $\Omega(\sqrt{m}/\log(m))$. 

    \item For all $k = \Omega(\log(m))$, $\MIRsubmod(m, k) = O(\sqrt{m/k})$. Moreover, our mechanism can be implemented with $2^k$ simultaneous value queries and computation, and is optimal with respect to the value query and computational/succinct representation models. The mechanism also works for bidders with subadditive valuations. When $k = \Theta(\log(m))$, this improves the previous best approximation guarantee for polynomial communication maximal-in-range mechanisms from $O(\sqrt{m})$~\cite{DobzinskiNS10} to $O(\sqrt{m/\log(m)})$. 
\end{itemize}

\noindent Let also $\MIRgen(m,k)$ denote the best approximation guarantee achievable by an MIR mechanism using $2^k$ communication between bidders with general valuations over $m$ items. Then
\begin{itemize}
    \item For all $k = \Omega(\log(m))$, $\MIRgen(m, k) = \Omega(m/k)$. When $k = \Theta(\log(m))$, this improves the previous best lower bound for polynomial communication maximal-in-range mechanisms from $\Omega(m/\log^2(m))$~\cite{DanielySS15} to $\Omega(m/\log(m))$. 

    \item For all $k = \Omega(\log(m))$, $\MIRgen(m, k) = O(m/k)$. Moreover, our mechanism can be implemented with $2^k$ simultaneous value queries and computation, and is optimal with respect to the value query and computational/succinct representation models. When $k = \Theta(\log(m))$, this improves the previous best approximation guarantee for polynomial communication maximal-in-range mechanisms from $O(m/\sqrt{\log(m)})$~\cite{HolzmanKMT04} to $O(m/\log(m))$. 
\end{itemize}

\end{abstract}

\newpage
\setcounter{tocdepth}{1}
\tableofcontents

\newpage

\section{Introduction} \label{section:Introduction}

In a \emph{combinatorial auction}, a central designer has a set of items $M \coloneqq [m]$, and bidders $N \coloneqq [n]$. Each bidder $i$ has a monotone valuation $v_i: 2^M \to \mathbb{R}_+$, unknown to the designer. The designer interacts with the bidders to produce an allocation $A = (A_1,\ldots, A_n)$ of the items (where $A_i \cap A_j = \emptyset$ for all $i \neq j$), and their goal is to select one maximizing the welfare, defined $\sum_{i \in N} v_i(A_i)$.

As an algorithmic resource allocation problem, combinatorial auctions are extremely well-studied -- see further discussion in \Cref{sec:related}. Combinatorial auctions are similarly well-studied in economic settings, where the bidders' incentives are now relevant. That is, while an efficient communication protocol suffices in a purely algorithmic setting, that protocol must also be incentive compatible and incentivize all bidders to follow it. Here, the designer may also charge each bidder $i$ a price $p_i$, and the bidder aims to optimize their utility: $v_i(A_i) - p_i$. 

When considering either desiderata separately, the state-of-affairs is well-understood. For example, with polynomial communication, a tight $\Theta(\sqrt{m})$-approximation for monotone valuations~\cite{LehmannOS02, NisanS06, BlumrosenN05a}, a tight $2$-approximation for subadditive valuations~\cite{Feige09, EzraFNTW19}, and a tight $e/(e-1)$-approximation for XOS valuations~\cite{Feige09,DobzinskiNS10} are known. Additionally, the optimal achievable guarantee for submodular valuations is known to lie in $[2e/(2e-1),e/(e-1)-10^{-5}]$~\cite{FeigeV10, DobzinskiV13}. However, these protocols are not incentive compatible. Similarly, the classical \emph{Vickrey-Clark-Groves} (VCG) mechanism is incentive compatible and finds the welfare-maximizing allocation~\cite{Vickrey61,Clarke71,Groves73}, but requires exponential communication for any of the above-referenced valuation classes.

As such, a central open problem within Economics and Computation is understanding the extent to which communication-efficient truthful mechanisms can match the approximation guarantees of communication-efficient (not necessarily incentive compatible) protocols. A key framework to tackle this agenda are \emph{maximal-in-range} (MIR) mechanisms. For example, state-of-the-art deterministic truthful mechanisms for monotone, subadditive, XOS, and submodular valuations are maximal-in-range~\cite{HolzmanKMT04, DobzinskiNS10}. Our main results settle the approximation guarantees achievable by MIR mechanisms. We overview this agenda and our results below.


\subsection{Maximal-in-Range Mechanisms}

MIR mechanisms leverage the VCG mechanism to trade off approximation guarantees for efficiency. In particular, the VCG mechanism is quite general, and implies a truthful mechanism that maximizes welfare over any set of (possibly unstructured) outcomes. Specifically, one can define any \emph{allocation bank} $\calA$ of allocations, and the VCG mechanism will truthfully optimize welfare \emph{over all allocations in $\calA$}. Formally:

\begin{theorem}[VCG Mechanism~\cite{Vickrey61,Clarke71,Groves73}]
    Let $\calA$ be a collection of allocations, and let $\calP$ be any communication protocol among the $n$ bidders to find the welfare-maximizing allocation in $\calA$. Then there is a deterministic truthful mechanism that selects a welfare-maximizing allocation in $\calA$ using $n+1$ black-box calls to $\calP$.

    The resulting mechanism is termed a \emph{Maximal-in-Range Mechanism} for $\calA$~\cite{NisanR07}.
\end{theorem}
\vsthm

MIR mechanisms therefore provide a structured algorithmic framework to design deterministic truthful mechanisms: one selects an allocation bank $\calA$ and designs a protocol $\calP$ to optimize over it. This framework induces a tradeoff between efficiency and optimality: richer allocation banks ensure that the welfare-maximizing allocation in $\mathcal{A}$ is a good approximation, but may require significant communication to optimize over. Smaller allocation banks are easier to optimize over, but may not contain allocations with good welfare.\footnote{Note that it is also possible for one allocation bank to strictly contain another and also be easier to optimize over, so this tradeoff is conceptual rather than literal.}

For the following reasons, there is significant interest in understanding the approximation guarantees achievable by MIR Mechanisms:
\begin{itemize}
    \item In all settings, the best-known approximation guarantees achieved by poly-communication deterministic truthful mechanisms are achieved by MIR mechanisms~\cite{HolzmanKMT04, DobzinskiNS10}. Moreover, this claim has held for the entire duration of the study of combinatorial auctions (that is, no poly-communication deterministic truthful mechanisms that outperform the best-known MIR mechanism at the time have ever been discovered).\footnote{Note that~\cite{DobzinskiN11} discover \emph{exponential-communication} non-MIR deterministic truthful mechanisms that outperform the best \emph{poly-communication} MIR mechanisms in multi-unit domains~\cite{DobzinskiN11}, but \emph{poly-communication} deterministic truthful mechanisms have never outperformed MIR mechanisms.} 
    
    \item All deterministic truthful mechanisms satisfying four natural properties are affine maximizers, a generalization of MIR mechanisms~\cite{LaviMN03}.\footnote{An affine maximizer also is defined by an allocation bank $\mathcal{A}$, scalars $\vec{c} \in \mathbb{R}_{\geq 0}^n$, and an adjustment $v_0:\mathcal{A} \rightarrow \mathcal{A}$. The affine maximizer selects an allocation $(A_1,\ldots, A_n)$ optimizing $v_0(A_1,\ldots, A_n) + \sum_{i \in N} c_i \cdot v_i(A_i)$. Affine maximizers are also truthful using the VCG payment scheme.} Moreover, if an affine maximizer guarantees an $\alpha$-approximation using $\beta$ communication on all submodular/XOS/subadditive/monotone valuations, then the MIR mechanism with the same allocation bank does so for the same valuation class as well (folkore -- see \Cref{sec:Folklore} for a proof).
    
    \item Some conjecture that indeed MIR mechanisms are the optimal deterministic mechanisms (and therefore settling the approximation guarantees of poly-communication MIR mechanisms will eventually settle the approximation guarantees of all poly-communication deterministic truthful mechanisms), although this conjecture remains far from settled. However, we show that MIR mechanisms indeed achieve optimal approximation guarantees for \emph{all} deterministic truthful mechanisms in the value query model and the computational/succinct representation model (see \Cref{table:Results}).
\end{itemize}

Indeed, MIR mechanisms have been studied since the start of Economics and Computation as a field~\cite{NisanR07}, and several works over the past two decades make significant progress understanding their strengths and limitations. Our main results close these gaps. We now state our main results, and afterwards discuss their context. In the following theorem statements (and the rest of this paper), define $\MIRsubmod(m, k)$ (and $\MIRgen(m,k)$, respectively) to be the optimal approximation guarantee that can be achieved by an MIR mechanism using at most $2^k$ communication between bidders with submodular (and general, respectively) valuations over $m$ items.

\begin{theorem}[restate=LowerBound,name=] \label{thm:LowerBound}
    For all $k = \Omega(\log(m))$, $\MIRsubmod(m, k) = \Omega(\sqrt{m/(k\log(m/k))})$ when $n = \Omega(\sqrt{m/k})$. In particular, the best possible approximation guarantee for submodular valuations by MIR mechanisms with $\poly(m)$ communication is $\Omega(\sqrt{m}/\log (m))$.
\end{theorem}
\vsthm

This improves prior work, beginning with the impossibility of $m^{1/6}$ with polynomial communication for MIR mechanisms for submodular valuations by~\cite{DobzinskiN11}, which was later improved to $m^{1/3}/\log^{2/3}(m)$ by~\cite{DanielySS15}. 

The particular constant of $1/2$ in the exponent is significant because we now know that the MIR mechanism of~\cite{DobzinskiNS10} achieving an $O(\sqrt{m})$-approximation using polynomial communication is essentially tight. Still, our second main result improves their guarantee slightly.

\begin{theorem}[restate=Subadditive,name=] \label{thm:Subadditive}
    For all $k = \Omega(\log(m))$, $\MIRsubmod(m, k) = O(\sqrt{m/k})$. In particular, our mechanism guarantees a $O(\sqrt{m/\log m})$ approximation in $\poly(m)$ communication.
    
    Moreover, the mechanism we construct can be implemented using $2^{O(k)}$ simultaneous value queries, or in time $2^{O(k)}$ in the succinct representation model. Our mechanism guarantees an $O(\sqrt{m/k})$-approximation for subadditive valuations as well. 
\end{theorem}
\vsthm

\Cref{thm:Subadditive} is a mild improvement over~\cite{DobzinskiNS10} (it saves a $\sqrt{\log (m)}$ factor). Still, we note that the MIR mechanism of~\cite{DobzinskiNS10} is exceptionally simple, and no better guarantee was previously known. Together, \Cref{thm:LowerBound,thm:Subadditive} nail down the achievable approximation guarantees for submodular valuations by MIR mechanisms with $\poly(m)$ communication up to a factor of $\Theta(\sqrt{\log(m)})$, exponentially improved over the prior gap of $\tilde{\Theta}(m^{1/6})$. Additionally, it is worth noting that in the value query and computational/succinct representation models, the MIR mechanism we construct for subadditive valuations is optimal for \emph{all} deterministic truthful mechanisms; see \cite{DobzinskiV16} and \Cref{sec:Folklore}.

We also consider general valuations. Here, nearly-tight bounds were previously known, and we improve both to be tight up to constant factors.

\begin{theorem}[restate=General,name=] \label{thm:General}
    For all $k = \Omega(\log(m))$, $\MIRgen(m,k) = \Theta(m/k)$. In particular, the best possible approximation guarantee for monotone valuations by MIR mechanisms with $\poly(m)$ communication is $\Theta(m/\log m)$. Moreover, the mechanism we construct can be implemented using $2^{O(k)}$ simultaneous value queries, or in time $2^{O(k)}$ in the succinct representation model.
\end{theorem}
\vsthm

When considering $\poly(m)$ communication, this gives logarithmic improvements for the previous best impossibility result of $\Omega(m/\log^2(m))$ by~\cite{DanielySS15}, and the previous best approximation guarantee of $O(m/\sqrt{\log(m)})$ by~\cite{HolzmanKMT04}.\footnote{Although the $\Theta(\sqrt{\log(m)})$ improvement in the upper bound seems small, the mechanism of~\cite{HolzmanKMT04} is actually far from optimal for large $k$; for example, if we allow $2^{O(m)}$ communication for a sufficiently small constant, then the mechanism of \cite{HolzmanKMT04} is only $O(\sqrt{m})$-approximate, whereas ours is $O(1)$-approximate.} Still, the minor improvements are significant as both bounds are now tight.

\Cref{table:Results} places our work alongside prior work for all three considered models. Our main results consider the communication model, but also imply minor improvements in the value query or computational model for free (or modest additional work). We contextualize the key takeaways from this table.
\begin{itemize}
    \item In the communication model with arbitrary monotone valuations, we improve both the state-of-the-art mechanism and lower bound, reducing the gap between them from $\Theta(\log^{1.5}(m))$ to $\Theta(1)$. In other words, $\Theta(m/\log(m))$ is the best achievable guarantee in $\poly(m)$ communication. We briefly overview technical highlights in Section~\ref{sec:highlights}.
    
    \item In the communication model with subadditive/XOS/submodular valuations, we improve the state-of-the-art lower bounds and slightly improve the state-of-the-art mechanism, reducing the gap between them from $\tilde{\Theta}(m^{1/6})$ to $\Theta(\sqrt{\log(m)})$. In other words, $\tilde{\Theta}(\sqrt{m})$ is the best achievable guarantee in $\poly(m)$ communication. We briefly overview technical highlights in Section~\ref{sec:highlights}.
    
    \item In the value queries model and computational model with either arbitrary monotone or subadditive/XOS/submodular valuations, we slightly improve state-of-the-art mechanisms. In the value queries model, these improvements match pre-existing lower bounds on \emph{any $\poly(m)$-query deterministic truthful mechanism}.\footnote{For arbitrary monotone, subadditive, and XOS valuations, our mechanisms even match pre-existing lower bounds on \emph{any $\poly(m)$-query deterministic algorithm}.} In the computational model, we further slightly improve prior-best lower bounds on \emph{any $\poly(m)$-time deterministic truthful mechanism} to match our MIR mechanisms. That is, we now know that MIR mechanisms achieve the optimal approximation guarantees among all $\poly(m)$-query deterministic truthful mechanisms, and all $\poly(m)$-time deterministic truthful mechanisms in the succinct representation model. These MIR mechanisms follow by observing that our new communication-efficient mechanisms can be implemented with $\poly(m)$ value queries (in fact, simultaneous value queries). Our slight improvement on lower bounds follows a similar outline as prior work, but is more careful with lower order terms.\footnote{Our computational lower bounds also use the stronger assumption of the randomized Exponential Time Hypothesis instead of the assumption $\textsc{RP} \ne \textsc{NP}$ used in prior work. However, this is essentially necessary if we care about lower order terms.}
\end{itemize}

\bgroup
\newcommand\T{\rule{0pt}{3.1ex}}       
\newcommand\B{\rule[-1.7ex]{0pt}{0pt}} 
\def\arraystretch{1.5}
\begin{table}[ht]
    \centering\small
    \begin{tabular}{C{4.8cm}|C{4.8cm} C{5.3cm}}
Communication & General & Subadditive/XOS/Submodular \B \\ \hline
Prior Best Mechanisms & $O(m/\sqrt{\log(m)})$~\cite{HolzmanKMT04} & $O(\sqrt{m})$~\cite{DobzinskiNS10} \T \\
Our MIR Mechanisms & \bm{$O(m/\log(m))$}~(Thm.~\ref{thm:General}) & \bm{$O(\sqrt{m/\log(m)})$}~(Thm.~\ref{thm:Subadditive}) \\
Our MIR Lower Bounds & \bm{$\Omega(m/\log(m))$}~(Thm.~\ref{thm:General}) & \bm{$\Omega(\sqrt{m}/\log(m))$}~(Thm.~\ref{thm:LowerBound}) \T \\
Prior MIR Lower Bounds & $\Omega(m/\log^2(m))$~\cite{DanielySS15} & $\Omega(m^{1/3}/\log^{2/3}(m))$~\cite{DanielySS15}
    \end{tabular}

    \bigskip
    
    \begin{tabular}{C{4.8cm}|C{4.8cm} C{5.3cm}}
Value Queries & General & Subadditive/XOS/Submodular \B \\ \hline
Prior Best Mechanisms & $O(m/\sqrt{\log(m)})$~\cite{HolzmanKMT04} & $O(\sqrt{m})$~\cite{DobzinskiNS10} \T \\
Our MIR Mechanisms & \bm{$O(m/\log(m))$}~(Thm.~\ref{thm:General}) & \bm{$O(\sqrt{m/\log(m)})$}~(Thm.~\ref{thm:Subadditive}) \\
Truthful Lower Bounds & $\downarrow$ & $\Omega(\sqrt{m/\log(m)})$~\cite{DobzinskiV16} \\
Algorithmic Lower Bounds & $\Omega(m/\log(m))$~\cite{BlumrosenN05a} & $\Omega(\sqrt{m/\log(m)})$ for Subadditive/XOS~\cite{DobzinskiNS10}
    \end{tabular}

    \bigskip
    
    \begin{tabular}{C{4.8cm}|C{4.8cm} C{5.3cm}}
Computation (Succ.~Rep.) & General & Subadditive/XOS/Submodular \B \\ \hline
Prior Best Mechanisms & $O(m/\sqrt{\log(m)})$~\cite{HolzmanKMT04} & $O(\sqrt{m})$~\cite{DobzinskiNS10} \T \\
Our MIR Mechanisms & \bm{$O(m/\log(m))$}~(Thm.~\ref{thm:General}) & \bm{$O(\sqrt{m/\log(m)})$}~(Thm.~\ref{thm:Subadditive}) \\
Our Truthful Lower Bounds & \bm{$\Omega(m/\log(m))$}~(Thm.~\ref{thm:GeneralComputationalLowerBound}) & \bm{$\Omega(\sqrt{m/\log(m)})$}~(Thm.~\ref{thm:SubmodularComputationalLowerBound}) \\
Prior Truthful Lower Bounds (for constant $\varepsilon > 0$) & $m^{1-\varepsilon}$~\cite{DanielySS15} & $\sqrt{m^{1-\varepsilon}}$~\cite{DobzinskiV16}
    \end{tabular}

    \caption{Summary of our results (bolded) compared to prior work (unbolded) when $k = \Theta(\log(m))$. All models referenced above are for deterministic mechanisms. See Table~\ref{table:ResultsGeneral} in \Cref{sec:MissingProofs} for the same table with general $k$.}
    
    \label{table:Results}
\end{table}
\egroup



\subsection{Technical Highlights} \label{sec:highlights}

Below, we overview one technical highlight from our algorithms, and one technical highlight from our lower bounds.

\medskip
\textbf{Technical Background: Prior Algorithms.}
The state-of-the-art MIR mechanism for general valuations~\cite{HolzmanKMT04} and for submodular valuations~\cite{DobzinskiNS10} are quite different. Parameterizing their approaches by $k$, the MIR mechanism of~\cite{HolzmanKMT04} partitions the items into $k$ chunks of size $m/k$, and considers $\calA$ to be the set of allocations that keeps together all items in the same chunk. Observe that there are only $2^k$ sets that any bidder might possibly receive in $\calA$, so optimizing over $\calA$ can be done with $2^k$ communication.

On the other hand,~\cite{DobzinskiNS10} considers $\calA$ to be the set of allocations that either give all items to the same bidder, or that gives each bidder a set of at most $O(k/\log(m/k))$ items. Again, there are only $2^k$ sets that any bidder might receive, so optimizing over $\calA$ can be done with $2^k$ communication.

\medskip
\textbf{Technical Highlight: Our Algorithms.}
Our algorithm for monotone valuations adds just one new idea to that of~\cite{HolzmanKMT04}: consider multiple partitions. Specifically, repeat the process $z$ times of partitioning $M$ into $k$ chunks, calling the chunks $B^{(\ell)}_s$ for $\ell \in [z]$ and $s \in [k]$. Let $\mathcal{A}$ be the set of allocations such that each bidder receives a set of the form $\bigcup_{s \in S} B^{(\ell)}_s$ for some $\ell \in [z]$ and $S \subseteq [k]$ (that is, each bidder receives a set that picks a single partition, and then is a union of chunks for that partition). Taking $z = 2^{O(k)}$ still requires just $2^{O(k)}$ communication. The main difference to~\cite{HolzmanKMT04} is that while any single partition can only give an $m/\sqrt{k}$-approximation, the best of $2^{O(k)}$ partitions improve the guarantee to $m/k$. Informally, this is because by taking many partitions, it is likely that for every set of $k$ items, there exists a partition where the items go to different chunks. This allows any $k$ items to be optimally allocated, so allocating the ``most important'' $k$ items will yield an $m/k$-approximation.

Our algorithm for subadditive/XOS/submodular valuations is in some sense more like~\cite{HolzmanKMT04} than~\cite{DobzinskiNS10}. Sticking exclusively to an approach like~\cite{DobzinskiNS10} is almost optimal, but doomed to lose a $\sqrt{\log(m/k)}$ factor due to the fact that we can only exhaust over sets of size $O(k/\log(m/k))$ in $2^k$ communication. On the other hand, sticking exclusively to an approach like~\cite{HolzmanKMT04} cannot guarantee an $o(m/k)$ approximation since at most $k$ bidders are given items, so significant changes are needed to leverage this approach. Our algorithm is as follows:
\begin{itemize}
    \item Take $z = 2^{O(k)}$ partitions of the items into $\sqrt{m/k}$ buckets of $\sqrt{mk}$ items (so there are $2^{O(k)}\sqrt{m/k}$ total buckets).
    
    \item Within each bucket, take $z' = 2^{O(k)}$ partitions of the items into $k$ chunks of $\sqrt{m/k}$ items. This induces sets of the form $C_{s',\ell',s,\ell}$, where $\ell \in [z]$ determines which bucketing we use, $s \in [\sqrt{m/k}]$ labels the bucket, $\ell' \in [z']$ determines which partition of that bucket into chunks we use, and $s' \in [k]$ labels the chunk.
    
    \item Finally, $\mathcal{A}$ denotes the set of allocations where each bidder either receives all the items $M$, or receives a set of the form $\cup_{s' \in S} C_{s',\ell',s,\ell}$ for some $S \subseteq [k]$ and $\ell', s, \ell$. In other words, each bidder chooses a bucket ($\ell$ chooses the bucketing and $s$ chooses a bucket in that bucketing), chooses a partition of that bucket into chunks, and then receives a subset of the chunks for the chosen partition of the chosen bucket. 
\end{itemize}

Observe that there are again only $2^{O(k)} \cdot \sqrt{mk} \cdot 2^{O(k)} \cdot 2^{k} = 2^{O(k)}$ ways to choose such a set that a bidder might get in $\calA$, so $\calA$ can be optimized over in $2^{O(k)}$ communication.

Broadly, the idea is to use subadditivity and binomial tail bounds to argue that a bucketing exists where we only need to allocate $k$ items optimally within each bucket to get a $\sqrt{m/k}$-approximation. Then we leverage our mechanism for general valuations within each bucket to allocate the $k$ items.

The key high-level technical takeaway we wish to emphasize is that exhausting over collections of large chunks of items, versus exhausting over collections of a few items, seems to be ``the right'' way to achieve optimal approximation guarantees for MIR mechanisms. This is because each partition into large chunks of items can achieve the desired approximation ratios for many configurations of bidders simultaneously, which allows us to overcome the barrier that there are too many configurations of bidders to try satisfying them by asking for small sets.

\medskip
\textbf{Technical Background: Prior Lower Bounds.}
Prior lower bounds on MIR mechanisms follow from an argument of the following form: (1), derive structure on $\calA$, using the fact that $\calA$ guarantees a good approximation, then (2), show that this structure embeds a hard communication problem. At a very high-level, the initial approach of~\cite{DobzinskiN11} could be described as using first-principles for (1), and then a non-trivial reduction from \textsc{SetDisjointness} for (2). The state-of-the-art approaches of~\cite{BuchfuhrerDFKMPSSU10, DanielySS15} instead use advanced machinery based on generalizations of the VC-dimension for (1), so that a trivial argument for (2) suffices.

In slightly more detail,~\cite{BuchfuhrerDFKMPSSU10, DanielySS15} find a set $N'$ of bidders and $M'$ of items such that $\calA$ must contain every possible allocation of items in $M'$ to bidders in $N'$ (in this case, we say that $\calA$ \emph{shatters} $(M', N')$). Then so long as the valuation class requires exponential communication to exactly optimize welfare, a communication lower bound of $2^{\Omega(\abs{M'})}$ follows immediately for (2), because when restricting attention to $M'$ and $N'$, $\calA$ considers all allocations and is exactly optimal.

\medskip
\textbf{Technical Highlight: Our Lower Bounds.}
Our lower bounds leverage \emph{some} of the advanced machinery developed in~\cite{BuchfuhrerDFKMPSSU10, DanielySS15} to understand the structure of any $\calA$ that achieves a good approximation, but stops short of going all the way to shattering. Indeed, the bounds in~\cite{DanielySS15} are tight (up to perhaps logarithmic factors) for approaches that insist on fully shattering some $(M', N')$. Instead, we leverage just enough structure to move towards a communication lower bound. For general valuations, the structure established in prior work actually suffices for a direct reduction from \textsc{SetDisjointness} that saves a $\log(m)$ factor. For submodular valuations, we derive a novel structure on $\calA$, but ultimately avoid a full shattering argument to save a $\tilde{\Theta}(m^{1/6})$ factor.

The key high-level takeaway we wish to emphasize is that our lower bounds improve over prior results by uncovering ``the right'' structure on $\calA$ to enable a simple-but-not-trivial communication lower bound, rather than pushing all the way towards fully shattering.


\subsection{Related Work} \label{sec:related}

We've previously discussed the most-related work to ours:~\cite{NisanR07} introduces the concept of MIR mechanisms, based off principles of the VCG mechanism~\cite{Vickrey61, Clarke71, Groves73}.~\cite{HolzmanKMT04} provides the first (and until our work, state-of-the-art) approximation for general valuations via an MIR mechanism, which is also the best previous truthful deterministic mechanism.~\cite{DobzinskiNS10} provides the first (and until our work, state-of-the-art) approximation for submodular valuations via an MIR mechanism, which is also the best previous truthful deterministic mechanism.~\cite{DobzinskiN11} provide the first lower bounds on $\poly(m)$-communication MIR mechanisms for submodular valuations.~\cite{DanielySS15}, building on tools developed in~\cite{BuchfuhrerDFKMPSSU10}, improve their bounds for submodular valuations and provide the first bounds for general valuations.

Beyond these directly-related works, there is a rich body of works on combinatorial auctions broadly. These works provide context for the study of MIR mechanisms specifically. For example,~\cite{LaviMN03} establish that all mechanisms satisfying four natural properties are affine maximizers (a generalization of MIR mechanisms that achieve identical approximation guarantees). Therefore, our results on MIR mechanisms also immediately bound approximation guarantees achievable by this class of mechanisms. There are, however, deterministic mechanisms that are not affine maximizers (posted-price mechanisms are one such example). As such, there is also an active body of research aiming to understand the approximation guarantees of deterministic truthful mechanisms with bounded communication. Our mechanisms are now the state-of-the-art deterministic truthful mechanisms with $\poly(m)$ communication for submodular, XOS, subadditive, and general valuations. On the other hand, there are significantly fewer lower bounds that hold for all deterministic truthful mechanisms. Specifically, the only such result for any of these four classes is a $4/3+\varepsilon$ lower bound for two XOS bidders~\cite{AssadiKSW20, BravermanMW18, Dobzinski16b}. 

The discussion of the previous paragraph considers a protocol to be truthful if it is an ex-post Nash equilibrium for bidders to follow it. That is, as long as every other bidder is following the protocol \emph{for some plausible valuations $\vec{v}_{-i}$}, it is in bidder $i$'s best interest to follow the protocol as well (for all $i$). One could instead seek mechanisms that are dominant strategy truthful: even if the other players use bizarre strategies that are not prescribed for any $\vec{v}_{-i}$, it is still in bidder $i$'s best interest to follow the protocol. On this front,~\cite{DobzinskiRV22} recently establish that no dominant strategy truthful mechanism can achieve an $m^{1-\varepsilon}$ approximation for general valuations in $\poly(m)$ communication. This means that, up to lower-order terms, the MIR mechanisms we develop are also optimal among dominant strategy truthful mechanisms (our mechanisms are dominant strategy truthful because they can be implemented using $\poly(m)$ simultaneous communication).

Finally, there is significant related work on the communication complexity of combinatorial auctions broadly, considering protocols (without incentives)~\cite{NisanS06, DobzinskiNS10,Feige09, FeigeV10, DobzinskiV13, DobzinskiNO14, Assadi17, EzraFNTW19}, deterministic truthful mechanisms~\cite{LehmannOS02}, and randomized truthful mechanisms~\cite{DobzinskiNS06, Dobzinski07, KrystaV12, Dobzinski16a,AssadiS19, AssadiKS21}. There is also significant related work on the computational complexity of combinatorial auctions broadly, again considering protocols without incentives~\cite{LehmannLN01, Vondrak08, MirrokniSV08}, and strong inapproximability results for truthful and computationally efficient mechanisms~\cite{Dobzinski11, DughmiV11, DobzinskiV12a, DobzinskiV12b, DobzinskiV16}.

\section{Preliminaries and Notation} \label{sec:prelim}

\subsection{Shattering}

When convenient, we may think of an allocation $A : M \to N \cup \{*\}$ as a function from items to bidders, where $*$ denotes an item not allocated to any bidder. As such, we may use notation such as $A \vert_{M'} \coloneqq (A_1 \cap M', \dots, A_n \cap M')$ and $\calA \vert_{M'} \coloneqq \{A \vert_{M'} : A \in \calA\}$ to denote the restriction of an allocation/allocation bank to the items $M'$. For allocation banks on disjoint sets of items $\calA$ and $\calB$, we may also use the notation $\calA \times \calB \coloneqq \{(A_1 \cup B_1, \dots, A_n \cup B_n) : A \in \calA, B \in \calA\}$ to denote an allocation bank where each combination of $A \in \calA$ and $B \in \calB$ is possible.

For a set of bidders $v_1, \dots, v_n : 2^M \to \mathbb{R}_+$, define $A^*(\vec{v})$ to be an optimal allocation, and let $\OPT(\vec{v})$ be the welfare under $A^*(\vec{v})$. For any $N' \subseteq N$, define $\OPT(\vec{v}, N')$ to be the welfare of bidders $N'$ under $A^*(\vec{v})$. Let $\MIR_\calA(\vec{v})$ be the welfare from an optimal allocation in $\calA$.

One concept that will repeatedly appear in our arguments is that of an allocation bank $\calA$ \emph{shattering} a collection of items/bidders.

\begin{definition}[Shattering]
    An allocation bank $\calA$ \emph{$d$-shatters} a pair $(M', N')$ if for all items $j \in M'$, there exists a set $T_j \subseteq N'$ with $\abs{T_j} = d$ such that $\bigtimes_{j \in M'} T_j \subseteq \calA \vert_{M'}$. That is, for each of the $d^{\abs{M'}}$ ways to allocate each item $j \in M'$ to a bidder in $T_j$, there exists an allocation in $\calA$ that allocates the items in $M'$ in this manner.

    If $\calA$ $\abs{N'}$-shatters $(M', N')$, we simply say that $\calA$ shatters $(M', N')$.
\end{definition}
\vsthm

Prior lower bounds of~\cite{BuchfuhrerDFKMPSSU10, DanielySS15} use this concept extensively, and eventually find a large set of items that are shattered. Our lower bound for submodular functions leverage this machinery for $d < \abs{N'}$ instead of $d = \abs{N'}$ to achieve a $\tilde{\Theta}(m^{1/6})$ improvement. This concept is also helpful for understanding intuitively how our mechanisms provide good approximation guarantees.

\subsection{Formal Statement of Models}

Our main results consider the communication model, where each player $i$ holds the valuation function $v_i(\cdot)$ and we consider only the communication cost of the protocol (for concreteness, in the blackboard model). Our new mechanisms (like the previous-best mechanisms) can be implemented simultaneously using only value queries. As such, these also imply results in the value query model, and the succinct representation model. In the succinct representation model, each player has a $v_i(\cdot)$ that can be represented by an explicit circuit of size at most $\poly(m)$. Because our main results are parameterized by $k$, we will further refer to the \emph{$2^k$-succinct representation model} as the case where each $v_i(\cdot)$ can be represented by an explicit circuit of size at most $2^{O(k)}$.

Additionally, we make the simplifying assumption that when considering the class of mechanisms that can be run in $2^{O(k)}$ communication/value queries/computation, all numbers are integers less than $2^{2^{O(k)}}$ (and therefore can be represented in $2^{O(k)}$ bits in a standard fashion). We make this assumption to avoid any strangeness with things like arithmetic, value queries representing arbitrary precision numbers, etc.

\section{An Optimal MIR Mechanism for General Valuations} \label{section:GeneralMechanism}

In this section, we will prove the upper bound in \Cref{thm:General} by constructing an $m/k$-approximate MIR mechanism which uses $2^{O(k)}$ communication. 

\begin{definition}[Chunking Mechanism]
    Let $B^{(1)}, \dots, B^{(z)} \in [t]^M$ partition $M$ into $t$ \emph{chunks} each. The allocation bank $\calA$ contains every allocation where each bidder gets a set of the form $C(S,\ell):=\bigcup_{j \in S} B^{(\ell)}_j$ for some $\ell \in [z], S \subseteq [t]$, and the \emph{chunking mechanism for $B^{(1)}, \dots, B^{(z)}$} is MIR over $\calA$.
\end{definition}
\vsthm

The prior state-of-the-art for general valuations is simply a chunking mechanism for a single partition into $k$ chunks of equal size~\cite{HolzmanKMT04} (i.e.~a chunking mechanism with $z=1$). They prove their mechanism guarantees an $m/\sqrt{k}$ approximation ratio for all monotone valuations (and this is tight -- no chunking mechanism with $z=1$ can guarantee better than $m/\sqrt{k}$).

Our only new idea is to instead consider a chunking mechanism for \emph{multiple} carefully chosen partitions which satisfy the following property.

\begin{definition}[$r$-Itemizing]
    A partition $B^{(\ell)}$ \emph{itemizes} a set $S$ if each chunk of $B^{(\ell)}$ contains at most one item of $S$. A list of partitions $B^{(1)}, \dots, B^{(z)}$ is \emph{$r$-itemizing} if for any set $S$ of size at most $r$, there exists $B^{(\ell)}$ which itemizes $S$.
\end{definition}

\begin{lemma} \label{lemma:rItemizingPartitionsExist}
    For all $r = \Omega(\log\log(m))$ and some $z = 2^{\Theta(r)}$, there exists a list of partitions $B^{(1)}, \dots, B^{(z)} \in [r]^M$ into $r$ chunks which is $r$-itemizing.
\end{lemma}
\begin{proof}
Suppose we randomly sample the partitions such that $B^{(1)}, \dots, B^{(z)} \in [r]^M$ are independent and uniformly random. Then for a fixed set $S$ of size $r$, the partition $B^{(\ell)}$ itemizes $r$ w.p. $r!/r^r = 2^{-\Theta(r)}$. Therefore, by independence, no partition itemizes $S$ w.p. at most $(1 - 2^{-\Theta(r)})^z = 2^{-2^{\Theta(r)}}$. By a union bound over the $\binom{m}{r} \leq 2^{r\log(m)} \leq 2^{2^{\Theta(r)}}$ sets of size $r$, $B^{(1)}, \dots, B^{(z)}$ is $r$-itemizing w.p. $> 0$. Thus, there exists a fixed list of partitions $B^{(1)}, \dots, B^{(z)}$ that is $r$-itemizing.
\end{proof}

\begin{theorem} \label{thm:GeneralMechanism}
    Let $k = \Omega(\log(m))$ and $z = 2^{\Theta(k)}$, and let $B^{(1)}, \dots, B^{(z)} \in [r]^M$ be a $(4k)$-itemizing list of partitions, which exists by \Cref{lemma:rItemizingPartitionsExist}. Then the chunking mechanism for $B^{(1)}, \dots, B^{(z)}$ is $m/k$-approximate and can be implemented using $2^{O(k)}$ communication. 

    Moreover, the mechanism can be implemented simultaneously with $2^{O(k)}$ value queries, and in time $2^{O(k)}$ in the $2^k$-succinct representation model.
\end{theorem}
\begin{proof}
The first step in our analysis for general valuations is similar to the analysis of~\cite{DobzinskiNS10} for subadditive valuations, which separately analyzes the bidders who receive many items versus few items in the optimal allocation.

Let $t = m/(2k)$. We will partition bidders into sets $N_0, N_1, \dots, N_t$ such that $i \in N_0$ if and only if $\abs{A^*_i(\vec{v})} > 2k$, and for all $s > 0$, $\sum_{i \in N_s} \abs{A^*_i(\vec{v})} \leq 4k$. Observe that the condition on all $N_s$ is possible because the bidders not in $N_0$ all get at most $2k$ items each.\footnote{That is, this partition can be created by first placing all bidders with $|A_i^*(\vec{v})|> 2k$ in $N_0$, and then greedily filling $N_i$ with remaining bidders without exceeding the cap of $4k$. Because each bidder not in $N_0$ gets at most $2k$ items, each non-empty bidder set will have at least $2k$ items.}

Let $\calA$ be the allocation bank which defines the chunking mechanism and observe that $\calA$ can allocate all items to a single bidder. Then since there are at most $m/(2k)$ bidders who get more than $2k$ items, $\MIR_\calA(\vec{v}) \geq (2k/m)\OPT(\vec{v}, N_0)$.

Now, observe that any set $S$ which is itemized by some $B^{(\ell)}$ is shattered by $\calA$, as the chunking mechanism can assign any combination of the chunks (and hence any combination of the items) to the bidders. Thus, $\calA$ shatters every set of size $4k$. Since each set of bidders $N_s$ gets at most $4k$ items in $A^*(\vec{v})$, $\MIR_\calA(\vec{v}) \geq \max_{s \in [t]} \OPT(\vec{v}, N_s) \geq (2k/m)\OPT(\vec{v}, N \setminus N_0$).

Therefore, we get that for all valuations $v_1, \dots, v_n$,
\[
    \MIR_\calA(\vec{v}) \quad \geq \quad \max\bigg\{\frac{2k}{m}\OPT(\vec{v}, N_0),\; \frac{2k}{m}\OPT(\vec{v}, N \setminus N_0)\bigg\} \quad \geq \quad \frac{k}{m} \OPT(\vec{v}) \enspace .
\]

\textbf{Communication and Computation.} Each bidder can only receive at most $2^{4k} z = 2^{\Theta(k)}$ possible sets in $\calA$, so optimizing over $\calA$ can be done with just $2^{\Theta(k)}$ simultaneous value queries per bidder. On the computation side, we will make use of the following lemma.

\begin{lemma} \label{lemma:WelfareMaximizationComputation}
    In the $2^{m'}$-succinct representation model with items $M' \coloneqq [m']$ and bidders $N' \coloneqq [n']$, a welfare-maximizing allocation can be found in time $2^{O(m')} \cdot n'$.
\end{lemma}
\begin{proof}
For $T \subseteq N'$ and $S \subseteq M'$, define $v_T(S)$ to be the optimal welfare for bidders $T$ and items $S$, and define $A_T(S)$ to be an optimal allocation of items $S$ to bidders $T$. Suppose the functions $v_{[i]}$ and $A_{[i]}$ are known for some $i \in [n'-1]$. Then $v_{[i+1]}$ and $A_{[i+1]}$ can be computed in $2^{O(m')}$ computation by brute forcing over all $2^{m'}$ sets $S \subseteq M$ and all of at most $2^{m'}$ allocations of $S$ between $v_{[i]}$ and $v_{i+1}$. Hence, we can iteratively compute $v_{N'}$ and $A_{N'}$ in $2^{O(m')} n'$ time, and the optimal allocation is $A_{N'}(M')$.
\end{proof}

Observe that to run our mechanism, we only need to solve $z$ welfare maximization problems over $4k$ ``items,'' where we interpret a chunk as a single item. Therefore, by \Cref{lemma:WelfareMaximizationComputation}, the total computation needed is $2^{O(k)} nz = 2^{O(k)}$.
\end{proof}

\begin{remark}
    Note that while we can run the above mechanism in polynomial time \emph{given} an $r$-itemizing list of partitions, we do not know how to \emph{explicitly} find such a list in polynomial time. Therefore, if we want an explicit mechanism, then we can only achieve an $m/\log(m)$ approximation in polynomial time w.h.p. by sampling a random list of partitions. Note that this is still stronger than a mechanism which achieves the desired approximation with constant probability/in expectation since not all truthful mechanisms can have their success probability amplified by repetition.

    A similar statement holds true for the subadditive mechanism in the next section.
\end{remark}

\section{An MIR Mechanism for Subadditive Valuations} \label{section:SubadditiveMechanism}

In this section, we will prove \Cref{thm:Subadditive} by giving a $\sqrt{m/k}$-approximate MIR mechanism for subadditive valuations which uses $2^{O(k)}$ communication.

The prior state-of-the-art for subadditive valuations asks each bidder for their valuation for all sets of size $O(k/\log(m/k))$, and the entire set of items~\cite{DobzinskiNS10}. Then we can either give a single influential bidder every item (which yields an approximation by observing there are not too many influential bidders), or we can give every bidder their $k/\log(m/k)$ favorite items (which yields an approximation by subadditivity).

Our mechanism for subadditive valuations deviates from this approach of asking for smaller sets, and is instead more closely related to our approach for general mechanisms.

\begin{definition}[Bucketing Mechanism] \label{def:BucketingMechanism}
    Let $B^{(1)}, \dots, B^{(z)} \in [t]^M$ partition $M$ into $t$ \emph{buckets} each, and for each $\ell \in [z]$, $s \in [t]$, and $T \subseteq N$, let $\calA^{(\ell)}_{s,T} \subseteq T^{B^{(\ell)}_s}$ be an allocation bank for items $B^{(\ell)}_s$ among bidders $T$. Also, let $A^{(i)}$ denote the allocation that awards all items to bidder $i$. Then the \emph{bucketing mechanism for $\{\calA^{(\ell)}_{s,T} : \ell \in [z], s \in [t], T \subseteq N\}$} is MIR over the allocation bank
    \[
        \calA \quad \coloneqq \quad \bigg(\bigcup_{i \in N} \{A^{(i)}\}\bigg) \cup \bigg(\bigcup_{\ell \in [z]} \bigcup_{P \in [t]^N} \bigtimes_{s \in [t]} \calA^{(\ell)}_{s,P_s}\bigg) \enspace .
    \]
    In other words, $\mathcal{A}$ includes all allocations that award all items to the same bidder. All other allocations in $\mathcal{A}$ first choose a bucketing $\ell \in [z]$, then partitions bidders among buckets (with bidders $P_s$ going to bucket $s$), and then finally chooses an allocation in $\mathcal{A}^{(\ell)}_{s, P_s}$ of items in bucket $B^{(\ell)}_s$ to bidders $P_s$. The bucketing mechanism for $\mathcal{A}$ is MIR over $\mathcal{A}$.
\end{definition}

\begin{definition}[Regular]
    A partition $B^{(\ell)} \in [t]^M$ is \emph{regular} for a (possibly incomplete) partition of the items $B$ into $t$ buckets if for all $s \in [t]$ where $\abs{B_s} = O(m/t)$, $\abs{B^{(\ell)}_{s'} \cap B_s} = O(m/t^2)$ for all $s' \in [t]$. A list of partitions $B^{(1)}, \dots, B^{(z)} \in [t]^M$ is \emph{regular} if for all $B$, some $B^{(\ell)}$ is regular for $B$.
\end{definition}

\begin{lemma} \label{lemma:RegularPartitionsExist}
    For $t = O(\sqrt{m/\log(m)})$, there exists a regular list of partitions $B^{(1)}, \dots, B^{(m)} \in [t]^M$.
\end{lemma}
\begin{proof}
Suppose we randomly sample the partitions such that $B^{(1)}, \dots, B^{(m)} \in [t]^M$ are independent and uniformly random. Then for a fixed (possibly incomplete) partition $B$ into $t$ buckets and any $s \in [t]$ where $\abs{B_s} \leq Cm/t$ for some constant $C$, $\abs{B^{(\ell)}_{s'} \cap B_s}$ is stochastically dominated by $X \sim \Binom(Cm/t, 1/t)$. Thus, noting that $m/t^2 = \Omega(\log(m))$, we have
\begin{align*}
    \Pr\Big[B^{(1)}, \dots, B^{(m)} \text{ is not regular for } B\Big] \quad &\leq \quad \Big(t^2 \Pr[X = \omega(m/t^2)]\Big)^m \\
    &\leq \quad \Big(mt^2 \Pr[X = 3Cm/t^2]\Big)^m \\
    &\leq \quad \bigg(mt^2 \binom{Cm/t}{3Cm/t^2} \bigg(\frac{1}{t}\bigg)^{3Cm/t^2}\bigg)^m \\
    &\leq \quad \bigg(mt^2 \bigg(\frac{Cem/t}{3Cm/t^2}\bigg)^{3Cm/t^2} \bigg(\frac{1}{t}\bigg)^{3Cm/t^2}\bigg)^m \\
    &= \quad \bigg(mt^2 \bigg(\frac{e}{3}\bigg)^{3Cm/t^2}\bigg)^m \\
    &\leq \quad \frac{1}{m^m} \enspace .
\end{align*}
Hence, by a union bound over $\leq 2^m t^m < m^m$ (possibly incomplete) partitions of the items into $t$ buckets, $B^{(1)}, \dots, B^{(m)}$ is regular w.p. $> 0$, and therefore there exists a fixed list of partitions $B^{(1)}, \dots, B^{(m)}$ which is regular.
\end{proof}

We now present a $\sqrt{m/k}$-approximate mechanism for subadditive valuations using $2^{O(k)}$ simultaneous value queries.

\begin{definition}[Bucket-Shattering Mechanism] \label{def:SubadditiveMechanism}
    Let $k = \Omega(\log(m))$, let $t = \sqrt{m/k}/2$, and let $B^{(1)}, \dots, B^{(m)} \in [t]^M$ be a regular list of partitions, which exists by \Cref{lemma:RegularPartitionsExist}. For each $\ell \in [z], s \in [t]$, fix a $\Theta(k)$-itemizing list of $2^{\Theta(k)}$ partitions for $B^{(\ell)}_s$ (which exists by Lemma~\ref{lemma:rItemizingPartitionsExist}), and for each $T \subseteq N$, let $\calA^{(\ell)}_{s,T}$ be its chunking mechanism (i.e.~the $\Theta(m/k)$-approximate MIR mechanism from \Cref{thm:GeneralMechanism} for this specific list of partitions). The \emph{bucket-shattering mechanism for $k$} is the bucketing mechanism for this choice of $\mathcal{A}^{(\ell)}_{s,P_s}$.
\end{definition}

\begin{example}
    Suppose we have $m = 8$ and $n = 4$. Then one bucketing and chunking is
    \begin{align*}
        \text{Buckets 1: } (\{1, 2, 3, 4\}, \{5, 6, 7, 8\}) &\quad \begin{cases}
            \text{Chunks 1a: } (\{1, 2\}, \{3, 4\}), \enspace \text{Chunks 2a: } (\{1, 3\}, \{2, 4\}) \\
            \text{Chunks 1b: } (\{5, 6\}, \{7, 8\}), \enspace \text{Chunks 2b: } (\{5, 7\}, \{6, 8\})
        \end{cases} \\
        \text{Buckets 2: } (\{1, 3, 5, 7\}, \{2, 4, 6, 8\}) &\quad \begin{cases}
            \text{Chunks 1a: } (\{1, 3\}, \{5, 7\}), \enspace \text{Chunks 2a: } (\{1, 5\}, \{3, 7\}) \\
            \text{Chunks 1b: } (\{2, 4\}, \{6, 8\}), \enspace \text{Chunks 2b: } (\{2, 6\}, \{4, 8\})
        \end{cases}
    \end{align*}

    Then we could
    \begin{itemize}[topsep=0pt,itemsep=0pt]
        \item Choose bucketing 1

        \item Choose chunking 1a for bucket a and chunking 2b for the bucket b.

        \item Assign bidder 1 to the bucket a and bidders 2, 3, 4 to the bucket b.

        \item Give bidder 1 chunks $\{1, 2\}$ and $\{3, 4\}$ from bucket a.

        \item Give bidder 2 chunk $\{5, 7\}$ and bidder 4 chunk $\{6, 8\}$ from bucket b.
    \end{itemize}
    This results in the allocation $(\{1, 2, 3, 4\}, \{5, 7\}, \emptyset, \{6, 8\})$. Any allocation resulting from a similar procedure would be in the allocation bank of the bucket-shattering mechanism.

    On the other hand, the allocation $(\{1, 2, 5, 6\}, \{3, 4\}, \{7, 8\}, \emptyset)$ is impossible, because while we could choose bucketing 1, and chunkings 1a for bucket a and 1b for bucket b, bidder 1 can only receive chunks from a single bucket. This is a crucial restriction that saves a large factor of communication.
\end{example}

\begin{remark}
    One can think of the bucket-shattering mechanism as adding an additional layer of shattering to the mechanism for general valuations: we first shatter the bidders among the buckets, in the sense that any allocation of the bidders to buckets is possible. Then we run the general mechanism within each bucket, which shatters the items among the bidders for that bucket.
    
    Note that once the general mechanism within each bucket is solved, no additional communication is needed to find the optimal bucketing. However, $2^{\Omega(n)}$ \emph{computation} is needed (at least naively), which is too much when $n = \omega(k)$. This can be avoided by restricting consideration to only $\poly(m)$ random allocations of bidders to buckets instead of all $t^n$ allocations. We first analyze the more elegant, computationally inefficient version.
\end{remark}

\begin{theorem} \label{thm:SubadditiveMechanismComputationallyInefficient}
    For $k = \Omega(\log(m))$, the bucket-shattering mechanism for $k$ (\Cref{def:SubadditiveMechanism}) is $\sqrt{m/k}$-approximate for subadditive valuations and can be implemented using $2^{O(k)}$ communication.
    
    Moreover, the mechanism can be implemented simultaneously with $2^{O(k)}$ value queries.
\end{theorem}
\begin{proof}
We again first partition bidders into sets $N_0, N_1, \dots, N_t$ such that $i \in N_0$ if and only if $\abs{A^*_i(\vec{v})} \geq 2\sqrt{mk}$, and for all $s \in [t]$, $\sum_{i \in N_s} \abs{A^*_i(\vec{v})} \leq 4\sqrt{mk}$.

Observe that $\calA$ can allocate all items to a single bidder, so since there are at most $\sqrt{m/k}/2$ bidders who get more than $2\sqrt{mk}$ items, $\MIR_\calA(\vec{v}) \geq 2\sqrt{k/m}\OPT(\vec{v}, N_0)$.

Define $B^*_s \coloneqq \bigcup_{i \in N_s} A^*_i(\vec{v})$, meaning that $B^*_s$ is the items that bidders in bucket $s$ get in an optimal allocation. By construction of $N_1, \dots, N_t$, $\abs{B^*_s} \leq 4\sqrt{mk} = 2m/t$ for all $s \in [t]$. Then if we interpret $B^* \coloneqq (B^*_1, \dots, B^*_t)$ as an incomplete partition into $t$ buckets, there exists $\ell \in [m]$ such that $B^{(\ell)}$ is regular for $B^*$.

Now, define $B^{(\ell)}_s$ for $s \in [t+1, 2t]$ by $B^{(\ell)}_{s-t}$, and observe that by subadditivity of $v_1, \dots, v_n$,
\[
    \sum_{\Delta \in [t]} \sum_{s \in [t]} \sum_{s \in N_s} v_i(B^{(\ell)}_{s+\Delta} \cap A^*_i(\vec{v})) \quad \geq \quad \sum_{s \in [t]} \sum_{i \in N_s} v_i(A_i^*(\vec{v})) \quad = \quad \OPT(\vec{v}, N \setminus N_0) \enspace ,
\]
because $\bigcup_{\Delta \in [t]} (B^{(\ell)}_{s+\Delta} \cap A^*_i(\vec{v})) = A^*_i(\vec{v})$ for any $s, i$. Hence,
\begin{equation} \label{eq:SubadditiveMainProperty}
    \max_{\Delta \in [t]} \bigg\{\sum_{s \in [t]} \sum_{i \in N_s} v_i(B^{(\ell)}_{s+\Delta} \cap A^*_i(\vec{v}))\bigg\} \quad \geq \quad 2\sqrt{\frac{k}{m}} \OPT(\vec{v}, N \setminus N_0) \enspace .
\end{equation}

Now, observe that:
\begin{itemize}
    \item $\bigtimes_{s \in [t]} \calA^{(\ell)}_{s+\Delta,N_s} \subseteq \calA$ for all $\Delta \in [t]$.

    \item $\calA^{(\ell)}_{s+\Delta,N_s}$ shatters every subset of $B^{(\ell)}_s$ of size $O(k)$ (see proof of \Cref{thm:GeneralMechanism}).

    \item $\sum_{i \in N_s} \abs{B^{(\ell)}_{s+\Delta} \cap A^*_i(\vec{v})} = \abs{B^{(\ell)}_{s+\Delta} \cap B^*_s} = O(k)$, since $B^{(\ell)}$ is regular for $B^*$.
\end{itemize}
By the above points, there exists $A \in \calA$ such that $B^{(\ell)}_{s+\Delta} \cap A^*_i(\vec{v}) \subseteq A_i$ for all $i \in N$ and $\Delta \in [t]$. Therefore, by \eqref{eq:SubadditiveMainProperty},
\[
    \MIR_\calA(\vec{v}) \quad \geq \quad 2\sqrt{\frac{k}{m}}\OPT(\vec{v}, N \setminus N_0) \enspace ,
\]
and thus,
\[
    \MIR_\calA(\vec{v}) \quad \geq \quad \max\bigg\{2\sqrt{\frac{k}{m}}\OPT(\vec{v}, N_0),\, 2\sqrt{\frac{k}{m}}\OPT(\vec{v}, N \setminus N_0)\bigg\} \quad \geq \quad \sqrt{\frac{k}{m}}\OPT(\vec{v}) \enspace .
\]

\textbf{Communication.}
Observe that for each $\ell \in [m]$, $s \in [t]$, and $T \subseteq N$, $\calA^{(\ell)}_{s,T} \subseteq \calA^{(\ell)}_{s,N}$, and each bidder can only receive at most $2^{\Theta(k)}$ possible sets in $\calA^{(\ell)}_{s,N}$, as it is a Chunking mechanism with $2^{\Theta(k)}$ partitions into $\Theta(k)$ chunks. Hence, each bidder can only receive at most $2^{\Theta(k)} mt = 2^{\Theta(k)}$ sets in $\calA$, so optimizing over $\calA$ can be done with just $2^{\Theta(k)}$ simultaneous value queries per bidder.
\end{proof}

\subsection{Computational Efficiency}

Naively, the bucket-shattering mechanism requires $2^{\Omega(n)}$ computation to implement, as there are exponentially many allocations of bidders to buckets. We can resolve this by randomly sampling polynomially-many allocations of bidders to buckets, and optimizing over this restricted subset instead. We define the mechanism below, but relegate all proofs to~\Cref{sec:MissingProofs}.

\begin{definition}[$\calP$-Bucketing Mechanism]
    Let $\calP \subseteq [t]^N$ be a set of partitions of the bidders into $t$ buckets, and for each $\ell \in [z]$, $s \in [t]$, and $T \subseteq N$, define $\calA^{(\ell)}_{s,T}$ as in \Cref{def:BucketingMechanism}. Similarly, let $A^{(i)}$ denote the allocation that awards all items to bidder $i$. Then the \emph{$\calP$-bucketing mechanism for $\{\calA^{(\ell)}_{s,T} : \ell \in [z], s \in [t], T \subseteq N\}$} is MIR over the allocation bank
    \[
        \calA \quad \coloneqq \quad \bigg(\bigcup_{i \in N} \{A^{(i)}\}\bigg) \cup \bigg(\bigcup_{\ell \in [z]} \bigcup_{P \in \calP} \bigtimes_{s \in [t]} \calA^{(t)}_{s,P_s}\bigg) \enspace .
    \]
    In other words, the $\calP$-bucketing mechanism is the bucketing mechanism with a restricted range for the assignment of bidders to buckets.
\end{definition}

\begin{definition}[restate=Balanced,name=Balanced] \label{def:Balanced}
    Let $\vec{v} \coloneqq (v_1, \dots, v_n)$, and let $N_1(\vec{v}) \coloneqq \{i \in N : \abs{A^*_i(\vec{v})} \leq m/t\}$. For any bucketing of the bidders $P \in [t]^N$ and $s \in [t]$, let $B^*_s(P) \coloneqq \bigcup_{i \in P_s \cap N_1(\vec{v})} A^*_i(\vec{v})$, let $S_P(\vec{v}) \coloneqq \{s \in [t] : \abs{B^*_s(\vec{v}, P)} = O(m/t)\}$, and let $N_P(\vec{v}) \coloneqq \bigcup_{s \in S_P(\vec{v})} (P_s \cap N_1)$. In other words, $N_P(\vec{v})$ is the set of bidders which belong to buckets that do not receive many items, when we restrict attention only to items awarded in OPT to bidders who do not individually receive many items (in OPT).
    
    Then a bucketing $P^{(\ell)} \in [t]^N$ is \emph{balanced} for $\vec{v}$ if $\OPT(\vec{v}, N_P(\vec{v})) = \Theta(\OPT(\vec{v}, N_1(\vec{v})))$. A list of bucketings $P^{(1)}, \dots, P^{(z)} \in [t]^N$ is \emph{balanced} if for all $\vec{v}$, some $P^{(\ell)}$ is balanced for $\vec{v}$.
\end{definition}

\begin{lemma}[restate=BalancedBucketingsExist,name=] \label{lemma:BalancedBucketingsExist}
    For $y = \poly(m)$, there exists a balanced list of bucketings $P^{(1)}, \dots, P^{(y)} \in [t]^N$.
\end{lemma}

\begin{definition}[restate=EfficientBucketingMechanism,name=Efficient Bucket-Shattering Mechanism] \label{def:EfficientBucketingMechanism}
    For $y = \poly(m)$, let $P^{(1)}, \dots, P^{(y)}$ be a balanced list of bucketings, which exists by \Cref{lemma:BalancedBucketingsExist}. Let $\calP$ be the set of all $P^{(x)}$ and their shifts (e.g., $(P^{(x)}_{1+\Delta}, \dots, P^{(x)}_{t+\Delta})$ for all $\Delta \in [t]$). The \emph{efficient bucket-shattering mechanism for $k$} is defined as the $\calP$-bucketing mechanism for $\{\calA^{(\ell)}_{s,T} : \ell \in [z], s \in [t], T \subseteq N\}$ as defined by the bucket-shattering mechanism for $k$.
\end{definition}

\begin{theorem}[restate=SubadditiveMechanismComputationallyEfficient,name=]
    For $k = \Omega(\log(m))$, the efficient bucket-shattering mechanism for $k$ (\Cref{def:EfficientBucketingMechanism}) is $\sqrt{m/k}$-approximate for subadditive valuations and can be implemented simultaneously with $2^{O(k)}$ value queries, and in time $2^{O(k)}$ in the $2^k$-succinct representation model.
\end{theorem}

\section{Communication Lower Bounds} \label{section:LowerBounds}

In this section, we will prove the lower bound in \Cref{thm:General} and \Cref{thm:LowerBound}.

Our approach can be broken down into the same two parts as prior approaches: we first show that guarantees on the approximation ratio implies a rich allocation bank, then show that optimizing over the allocation bank requires lots of communication. While the first part primarily adapts existing results, the second part uses novel techniques to close the $m^{1/6}$ gap between the upper and lower bounds for submodular valuations.


\subsection{Part 1: Approximation Implies Rich Allocation Bank}

We will show that a rich allocation bank $2$-shatters a large set of items. To do so, we make use of the following results.

\begin{proposition}[\cite{DanielySS15}, Theorem 1.5] \label{prop:MultiVCDimension}
    Suppose an allocation bank $\calA \subseteq N^M$ does not $2$-shatter any set of size $d$. Then
    \[
        \abs{\calA} \quad \leq \quad \sum_{i=0}^d \binom{m}{i} \binom{n}{2}^i \quad \leq \quad \bigg(\frac{emn^2}{d}\bigg)^d \enspace .
    \]
    For an allocation bank $\calA$, there exists $S \subseteq M$ such that $\calA$ $2$-shatters $(S, N)$ and
    \[
        \abs{\calA} \quad \leq \quad (mn^2)^{\abs{S}} \enspace .
    \]
\end{proposition}

\begin{proposition}[\cite{BuchfuhrerDFKMPSSU10}, Lemma 3.2] \label{prop:SubadditiveManyAllocations}
    If the MIR mechanism for $\calA$ is $n/3$-approximate for additive valuations, then there exists $S \subseteq M$ such that $\abs{\calA \vert_S} \geq 2^{m/n}$.
\end{proposition}

\begin{lemma} \label{lemma:RichAllocationsSubmodular}
    Let $k = \Omega(\log(m))$ and $n = 3\sqrt{m/k}$. Then if the MIR mechanism for $\calA$ is $\sqrt{m/k}$-approximate for submodular valuations, there exists $S \subseteq M$ of size $\Omega(\sqrt{mk}/\log(m/k))$ such that $\calA$ $2$-shatters $(S, N)$.
\end{lemma}
\begin{proof}
Suppose that $\calA$ does not $2$-shatter any set of size $\Omega(\sqrt{mk}/\log(m/k))$. Then by \Cref{prop:MultiVCDimension}, choosing the correct constants yields
\[
    \abs{\calA} \quad \leq \quad \bigg(\frac{emn^2}{\Theta(\sqrt{mk}/\log(m/k))}\bigg)^{\Theta(\sqrt{mk}/\log(m/k))} \quad = \quad O\bigg(\frac{m^2}{k^2}\bigg)^{\Theta(\sqrt{mk}/\log(m/k))} \quad < \quad 2^{\sqrt{mk}/3} \enspace .
\]
On the other hand, by \Cref{prop:SubadditiveManyAllocations}, $\abs{\calA} \geq 2^{\sqrt{mk}/3}$, a contradiction. Thus, $\calA$ must $2$-shatter some set of size $\Omega(\sqrt{mk}/\log(m/k))$.
\end{proof}

\begin{lemma} \label{lemma:RichAllocationsGeneral}
    Let $k = \Omega(\log(m))$ and $n \geq 2m/k$. Then if the MIR mechanism for $\calA$ is $m/k$-approximate for all monotone valuations, there exists $S \subseteq M$ of size $\Omega(k)$ and $T \subseteq N$ of size $2$ such that $\calA$ shatters $(S, T)$.
\end{lemma}
\begin{proof}
Suppose for simplicity that $n = 2m/k$ is an integer and $k$ is even. Let $\calB \subseteq N^M$ be the collection of partitions such that each part has $k/2$ items. The valuations \emph{induced} by $B \in \calB$ are $v_i(R) \coloneqq \mathbbm{1}_{R \supseteq B_i}$. Consider all partitions $B \in N^M$ where $\abs{B_i} = k/2$ for all $i \in N$, and consider the valuations $v_i(S) \coloneqq \mathbbm{1}_{S \supseteq B_i}$ induced by each such $B$.

Since the MIR mechanism for $\calA$ is $m/k$-approximate, for every set of valuations induced by $B \in \calB$, there exist two bidders who get value $1$. Since there are only $\binom{n}{2} < m^2$ pairs of bidders, there exists a fixed pair of bidders $T = \{i, j\}$ such that the MIR mechanism for $\calA$ gives bidder $i$ items $B_i$ and bidder $j$ items $B_j$ for a $1/m^2$ fraction of the $B \in \calB$.

Now, partition $\calB$ into $\{\calB_R : R \subseteq M,\, \abs{R} = k\}$ so that $B \in \calB_S$ if $B_i \cup B_j = R$. It follows that there exists some $\calB_R$ such that the MIR mechanism for $\calA$ gives bidder $i$ items $B_i$ and bidder $j$ items $B_j$ for a $1/m^2$ fraction of the $B \in \calB_R$.

Therefore, there exists $R \subseteq M$ of size $k$ such that $\abs{T^R \cap \calA \vert_R} \geq \abs{\calB_R \vert_R}/m^2 = \binom{k}{k/2}/m^2 = 2^{\Omega(k)}$. Applying Sauer's Lemma to $T^R \cap \calA \vert_R$, there exists $S$ of size $\Omega(k)$ such that $\calA$ shatters $(S, T)$.
\end{proof}

\begin{remark}
    We have demonstrated that good approximations in either the submodular or arbitrary monotone valuation setting imply a rich allocation bank in the sense that there is a large set of $2$-shattered items. This unified view then allows us (in the next section) to derive, from $2$-shattering, a useful structure for a lower bound, which then implies lower bounds for both submodular and arbitrary monotone valuations. However, we note that enough structure is already recovered by \Cref{lemma:RichAllocationsGeneral} to give a direct lower bound for monotone valuations, e.g., by noting that it requires $2^{\Omega(k)}$ communication to maximize the sum of two monotone functions over the $k$ items shattered between $2$ bidders~\cite{NisanS06}.
\end{remark}


\subsection{Part 2a: Rich Allocation Bank Contains Structure}

For simplicity, let $S \coloneqq [s]$ and $T \coloneqq [t]$. Our aim now is to find a structure within the $2$-shattered $(S, T)$ which will be suitable for embedding \textsc{SetDisjointness}. Let $T_j$ for $j \in S$ be the pair of bidders that item $j$ can go to. In other words, $\bigtimes_{j \in S} T_j$ is the allocation bank witnessing the $2$-shattering. Our end goal will be to prove the following:

\begin{claim} \label{claim:Structure}
    For $t = O(s/\log(s))$ and $z = 2^{\Theta(s/t)}$, there exist $B^{(1)}, \dots, B^{(z)} \in \bigtimes_{j \in S} T_j$ such that
    \begin{itemize}[topsep=0pt,itemsep=0pt]
        \item There exists $V \subseteq T$ such that any $f : V \to [z]$ is constant if and only if $B^{(f(1))}_1, \dots, B^{(f(t))}_t$ are pairwise disjoint.
        
        \item $\abs{B^{(1)}_i} = \dots = \abs{B^{(z)}_i}$ for all $i \in T$.
    \end{itemize}
\end{claim}

The next subsection makes clear why this structure allows us to embed \textsc{SetDisjointness}. At a high level, the first bullet ensures that we can encode valuations for the bidders using sets $X_1, \dots, X_t$ such that a certain welfare can be attained if and only if $\bigcap_{i \in V} X_i \ne \emptyset$. The first bullet alone suffices to give a lower bound for XOS valuations, and captures the main idea behind the lower bound. The second bullet introduces highly non-trivial technical challenges, but is required to extend the lower bound to submodular valuations. We only provide the proof of the claim in its entirety, but will comment on when simplifications can be made by not satisfying the second bullet.

\bigskip
\textbf{Proof of \Cref{claim:Structure}.} The first step is (with a slight abuse of notation) to consider $\calA \subseteq \bigtimes_{j \in S} T_j$ such that there exist $a_1, \dots, a_t$ such that for all $A \in \calA$, $\abs{A_i} = a_i$ for all $i$. In other words, $\calA$ satisfies the second bullet. Observe that there are $(s+1)^t \leq 2^{\varepsilon s}$ possible values for $a_1, \dots, a_t$ for any constant $\varepsilon > 0$, and hence there exists $\calA$ such that $\abs{\calA} \geq 2^{-\varepsilon s}\abs{\bigtimes_{j \in S} T_j} = 2^{(1-\varepsilon)s}$. Thus, we can achieve the second bullet without losing too many allocations from our bank.

\medskip
\textbf{Step I: A Helpful Interpretation.} Intuitively, we may think of the first bullet as saying that there exist $z$ allocations such that there is no way to combine two or more allocations into another valid allocation. In other words, the function $f$ chooses which allocation each bidder in $V$ receives a set from, and $f$ will always cause some item to be allocated to two bidders unless $f$ assigns each bidder to the same allocation.

Because each item can only go to one of two bidders, only items that go between bidders who are assigned different allocations under $f$ are able to cause the allocation induced by $f$ to be invalid. As such, it is helpful to interpret the $2$-shattering structure as a graph, and reason about the viability of mixing allocations across cuts in the graph.

\begin{definition}
    For $V \subseteq T$ and $E \subseteq S$ such that $\bigcup_{j \in E} T_j \subseteq V$, define $G(V, E)$ to be the graph where each bidder denotes a vertex and each item $j \in E$ denotes an edge between bidders $T_j$.
\end{definition}
\vsthm

Our eventual goal is to sample $z$ allocations independently and uniformly from $\calA$, and show that the probability that the first bullet is satisfied is nonzero. Had we instead considered $\bigtimes_{j \in S} T_j$ rather than $\calA$ (i.e., disregarded the second bullet), then observe that sampling uniformly from $\bigtimes_{j \in S} T_j$ assigns each item $j$ to a bidder in $T_j$ independently and uniformly, and hence the probability that independently sampled allocations can be combined is exponentially small in the number of edges crossing the cut separating the bidders who are assigned to different allocations. Then, it roughly suffices to take a subgraph with min-cut $\Omega(k)$, and apply some clever union bounds.

However, when we instead consider $\calA$, restricting to a $2^{-\varepsilon s}$ fraction of the original $2$-shattered structure introduces correlations between items when we sample uniformly from $\calA$, and hence a more complex argument is needed. Importantly, the analysis focuses on the following quantity.

\begin{definition}
    For a graph $G \coloneqq G(V, E)$, edges $C \subseteq E$ (typically, we will take $C$ to be the edges across a cut), and allocation bank $\calA \subseteq \bigtimes_{j \in E} T_j$, define $p(G, C, \calA)$ to be the maximum fraction of allocations that remain after fixing the allocation of the items in $C$. In other words,
    \[
        p(G, C, \calA) \quad \coloneqq \quad \max_{A \in \calA \vert_C} \frac{\abs{\{B \in \calA : A = B \vert_C\}}}{\abs{\calA}} \enspace .
    \]
\end{definition}
\vsthm

Observe that $p(G, C, \bigtimes_{j \in E} T_j) = 2^{-\abs{C}}$. To eventually apply a similar probabilistic argument as we would apply to a $2$-shattered structure, we will try to find a subgraph $H \coloneqq H(V, E)$ of $G(T, S)$ such that the min-cut of $H$ is $\Omega(k)$ and $p(H, C, \calA \vert_E) \approx 2^{-\abs{C}}$ for all cuts $C$.

\medskip
\textbf{Step II: Finding a Good Subgraph.} From here on, we will denote cuts by the set of their edges. Additionally, for a graph $G$, let $\gamma_r(G)$ denote its $r$-way min-cut, i.e., the smallest cut which partitions the bidders into $r$ non-empty parts.

\begin{lemma} \label{lemma:NiceGraph}
    Let $\calA \subseteq \bigtimes_{j \in S} T_j$ and $\abs{\calA} \geq 2^{(1-\varepsilon)s}$ for some constant $\varepsilon > 0$. Then there exists $V \subseteq T$ where $\abs{V} \geq 2$ and $E \subseteq S$ such that $H \coloneqq H(V, E)$ satisfies $\gamma_r(H) \geq \varepsilon(r-1)s/t$ for all $r$, and for any edges $C$ across a cut in $H$, $p(H, C) \leq 2^{-(1-2\varepsilon)s}$.
\end{lemma}
\begin{proof}
Initialize $S' \coloneqq S$, and while there exists an $r$-way cut $C$ across some connected component of $G \coloneqq G(T, S')$ such that $p(G, C, \calA \vert_{S'}) > 2^{-\varepsilon(r-1)s/t} 2^{-(1-2\varepsilon)\abs{C}}$, set $S' \coloneqq S' \setminus C$. Observe that by definition of $p(G, C, \calA \vert_{S'})$, we have $\abs{\calA \vert_{S' \setminus C}} \geq p(G, C, \calA \vert_{S'}) \abs{\calA \vert_{S'}}$.

Suppose for contradiction that the process results in a disconnected graph after removing the $r_1, \dots, r_\ell$-way cuts $C_1, \dots, C_\ell$. Then since removing an $r$-way cut increases the number of connected components by $r-1$, we have $\sum_{i \in [\ell]} (r_i-1) < t$. Further, there are no loops in $G$, so $S' = \emptyset$ and $\sum_{i \in [\ell]} \abs{C_i} = s$. Therefore,
\[
    1 \quad = \quad \abs{\calA \vert_\emptyset} \quad > \quad 2^{-\varepsilon s/t \sum_{i \in [\ell]} (r_i - 1)} 2^{-(1-2\varepsilon)\sum_{i \in [\ell]} \abs{C_i}} \abs{\calA} \quad > \quad 2^{-\varepsilon s} 2^{-(1-2\varepsilon)s} 2^{(1-\varepsilon)s} \quad = \quad 1 \enspace ,
\]
a contradiction. Thus, the process terminates with a connected component $H \coloneqq H(V, E)$ (which is a subgraph of the original $G(T, S)$) such that $\abs{V} \geq 2$ and every $r$-way cut $C$ across $H$ satisfies $p(H, C, \calA \vert_E) \leq 2^{-\varepsilon(r-1)s/t} 2^{-(1-2\varepsilon)\abs{C}} \leq 2^{-(1-2\varepsilon)\abs{C}}$. Additionally, $p(H, C, \calA \vert_E) \geq 2^{-\abs{C}}$, so the first exponential ensures that $\abs{C} \geq \varepsilon(r-1)s/t$ for all $C$. Thus, $H$ satisfies the desired properties.
\end{proof}

\medskip
\textbf{Step III: A Probabilistic Construction.} Before we proceed to the final probabilistic construction, we prove a property of the subgraph $H$ promised by \Cref{lemma:NiceGraph} that makes clear why we need a structure like $H$ to be contained in $G(T, S)$.

\begin{lemma} \label{lemma:BadCut}
    Let $H \coloneqq H(V, E)$ satisfy the condition in \Cref{lemma:NiceGraph} for $\varepsilon = 1/50$. Let $f : V \to [r]$ map bidders to parts of an $r$-way cut $C$. Then for independent $B^{(1)}, \dots, B^{(r)} \sim \calA \vert_E$, the probability that $\{B^{(f(i))}(i) : i \in V\}$ are pairwise disjoint is at most $2^{-\abs{C}/3}$.
\end{lemma}
\begin{proof}
Let $V_1, \dots, V_r$ be the partition of $V$ where $f(i) = \ell$ for any bidder $i \in V_\ell$. Let $C_\ell$ for $\ell \in [r]$ be the $2$-way cut between $V_\ell$ and $V \setminus V_\ell$.

Observe that there are $3^{\abs{C}}$ $r$-tuples of $B^{(1)} \vert_{C_1}, \dots, B^{(r)} \vert_{C_r}$ such that $B^{(f(i))}(i)$ for $i \in V$ are pairwise disjoint, because for each $j \in C$ which can go to either $V_\ell$ or $V_{\ell'}$, we need either $B^{(\ell)}$ to allocate $j$ to $V_{\ell'}$, or $B^{(\ell')}$ to allocate $j$ to $V_\ell$, or both.

On the other hand, by definition of $p(H, C_\ell, \calA \vert_E)$ and the upper bound on it given by \Cref{lemma:NiceGraph}, the probability of any such $r$-tuple being sampled is at most
\[
    \prod_{\ell \in [r]} p(H, C_\ell, \calA \vert_E) \quad \leq \quad \prod_{\ell \in [r]} 2^{-(1-2\varepsilon)\abs{C_\ell}} \quad = \quad 2^{-2(1-2\varepsilon)\abs{C}} \enspace ,
\]
so the probability that the $B^{(f(i))}(i)$ are pairwise disjoint is at most $3^{\abs{C}} 2^{-2(1-2\varepsilon)\abs{C}} \leq 2^{-\abs{C}/3}$.
\end{proof}

In other words, $H$ has the property that for every $r$-way cut in $H$, the probability that $r$ independently sampled allocations from $\calA \vert_E$ can be combined into a feasible allocation is exponentially small in the size of the cut, which is exactly the property we wished to replicate from the $2$-shattered structure! We can now proceed to the probabilistic construction. First, we give a graph theoretic result that will be needed for the union bound (proof in \Cref{sec:MissingProofs})

\begin{lemma}[restate=CutCounting,name=] \label{lemma:CutCounting}
    For any graph $G = (V, E)$ (possibly with parallel edges, but no loops) and $c \in \mathbb{Z}_{\geq 1}$, the number of $r$-way cuts with at most $c\gamma_2(G)$ edges is at most $\abs{V}^{4c}$.
\end{lemma}
\vsthm

\begin{lemma} \label{lemma:ProbabilisticStructure}
    Let $t = O(s/\log(s))$, $z = 2^{\Theta(s/t)}$, and $B^{(1)}, \dots, B^{(z)} \sim \calA \vert_E$. Then w.p. $> 0$, a function $f : V \to [z]$ is constant if and only if $\{B^{(f(i))}_i : i \in V\}$ are pairwise disjoint.
\end{lemma}
\begin{proof}
Since $B^{(1)}, \dots, B^{(z)}$ are valid allocations, the forward direction holds trivially.

For the reverse direction, fix any function $f : V \to [z]$ which takes on a fixed set of $r \geq 2$ values, and let the $C$ be the $r$-way cut induced by $f$. Then by \Cref{lemma:BadCut}, the probability that $\{B^{(f(i))}(i) : i \in V\}$ are pairwise disjoint is at most $2^{-\abs{C}/3}$.

By \Cref{lemma:CutCounting}, the number of distinct cuts $C$ with $\abs{C} \leq c\varepsilon s/t \leq c\gamma_2(H)$ is at most $t^{4c}$. Thus, by a union bound over all possible values of $r$, all possible fixed subsets of $r$ values in $[z]$, all possible values of $c$, and all possible cuts satisfying these parameters, the probability that the desired condition is not satisfied is at most
\begin{align*}
    \sum_{r=2}^t \bigg(\binom{z}{r} \sum_{c=r-1}^s t^{4(c+1)} 2^{-(c\varepsilon s/t)/3}\bigg) \quad &\leq \quad \sum_{r=2}^t \bigg(z^r \sum_{c=r-1}^s 2^{8c(\log t-\varepsilon s/(6t))}\bigg) \\
    &\leq \quad \sum_{r=2}^t \bigg(2^{\Theta(rs/t)} \sum_{c=r-1}^s 2^{-\Theta(cs/t)}\bigg) \\
    &\leq \quad \sum_{r=2}^t 2^{-\Theta(rs/t)} \\
    &< \quad 1 \enspace ,
\end{align*}
where the second line follows because $\log t \leq \log s = O(s/t)$.
\end{proof}

Extending the allocations $B^{(1)}, \dots, B^{(z)} \in \calA \vert_E$ to any allocations in $\calA$ which agree on the allocation of the items $E$ completes the proof of \Cref{claim:Structure}.


\subsection{Part 2b: Structure Yields Set Disjointness Embedding}

We will now show that the MIR mechanism for $\calA$ is capable of solving \textsc{SetDisjointness}. The bidder valuations we will use are the following.

\begin{definition} \label{def:MildDesires}
    A \emph{mild-desires bidder for $\calF \subseteq 2^M$}, where all $F \in \calF$ are the same size $a$, has the valuation function
    \[
        v(G) \quad = \quad \begin{cases}
            2\abs{G} & \abs{G} < a \\
            2\abs{G} - \mathbbm{1}_{G \not\in \calF} & \abs{G} = a \\
            2a & \abs{G} > a
        \end{cases} \enspace .
    \]
    We say that such a bidder is \emph{satisfied} if their allocation gives them value $2a$, which occurs when they receive items $F \in \calF$, or any $a+1$ items. Mild-desires bidders have submodular valuations~\cite{NisanS06}.
\end{definition}

For each bidder $i \in V$, associate an input set $X_i \subseteq [z]$. For the bidders $i \not\in V$, associate the set $X_i = [z]$. Fix $B^{(1)}, \dots, B^{(z)} \in \calA$ promised by \Cref{claim:Structure}. By bullet two, we can let each bidder $i \in T$ be mild-desires for $\{B^{(\ell)}_i : \ell \in X_i\}$.

Observe that if we were only interested in a communication lower bound for MIR mechanisms for XOS valuations, then we would not have needed bullet two to hold, as we could instead let the valuation for bidder $i$ be the rank function of the downward-closed set family defined by $\{B^{(\ell)}_i : \ell \in X_i\}$ (which is an XOS function).

Then to solve \textsc{SetDisjointness} with the MIR mechanism for $\calA$, we only need to show that $\bigcap_{i \in V} X_i \ne \emptyset$ if and only if the optimal welfare over $\calA$ is $2s$. The forward direction holds because the allocation $B^{(x)}$ for some $x \in \bigcap_{i \in V} X_i$ satisfies every bidder.

For the reverse direction, by \Cref{claim:Structure}, we know that any collection of sets desired by the bidders $V$ are pairwise disjoint (i.e., results in a valid allocation) if and only if those collection of sets belong to the same allocation. To achieve $2s$ welfare, we need to satisfy every bidder (since $\sum_{i \in T} a_i = s$), and this can only be done if every bidder receives a desired set (if we satisfy bidder $i$ by giving them $a_i + 1$ items, some bidder $j$ can only receive at most $a_j - 1$ items and cannot be satisfied). Hence, if the optimal welfare is $2s$, it must be the case that $\bigcap_{i \in V} X_i \ne \emptyset$.

Thus, the MIR mechanism for $\calA$ is capable of solving \textsc{SetDisjointness} over a universe of size $z$. Since the communication complexity of \textsc{SetDisjointness} is $\Omega(z)$~\cite{Nisan02}, we conclude that maximizing the welfare of submodular valuations over $\calA$ requires $2^{\Omega(s/t)}$ communication.\footnote{Note that because the randomized communication complexity of \textsc{SetDisjointness} is also $\Omega(b)$, even randomized protocols which maximize the welfare of monotone submodular valuations over $\calA$ must use $2^{\Omega(s/t)}$ communication in expectation.}

\medskip
\textbf{Wrapping Up.}
By \Cref{lemma:RichAllocationsGeneral,lemma:RichAllocationsSubmodular},
\begin{itemize}
    \item An $m/k$-approximate MIR mechanism for general valuations must use $2^{\Omega(k)}$ communication.

    \item A $\sqrt{m/k}$-approximate MIR mechanism for submodular valuations must use $2^{\Omega(k/\log(m/k))}$ communication, so a $\sqrt{m/(k\log(m/k))}$-approximate MIR mechanism for submodular valuations must use $2^{\Omega(k)}$ communication.
\end{itemize}
Thus, $\MIRgen(m, k) = \Omega(m/k)$ and $\MIRsubmod(m, k) = \Omega(\sqrt{m/(k\log(m/k))})$.

\section{Conclusion} \label{section:Conclusion}

For all amounts of communication, we improve both upper and lower bounds for approximation guarantees of MIR mechanisms over submodular, XOS, subadditive, and general valuations. This resolves the approximation guarantees of MIR mechanisms for general valuations up to a constant factor, and submodular/XOS/subadditive valuations up to a $\Theta(\sqrt{\log m})$ factor. In addition, the mechanisms which witness the upper bounds use only value queries, demonstrating that using arbitrary communication instead of the far more restrictive regime of value queries does not give a mechanism much power. Even so, there are a few open questions for future work.

\medskip
\textbf{Closing the Logarithmic Gap for Submodular Valuations.}
Although we were able to significantly improve existing lower bounds for submodular valuations (reducing the gap from $\tilde{\Theta}((m/k)^{1/6})$ to $\Theta(\sqrt{\log(m/k)})$), we still started from the same $2$-shattering argument of~\cite{DanielySS15} in order to embed a hard communication game.

If one conjectures that our lower bound can be slightly improved (as we do), this unfortunately cannot follow after a $2$-shattering argument -- there exist mechanisms which are $\sqrt{m/k}$-approximate (see \Cref{section:LogMGap}) and which do not $2$-shatter any pair $(S, T)$ such that $\abs{S}/\abs{T} = \Omega(k)$. 

If instead one conjectures that our MIR mechanisms can be slightly improved, then the mechanism must be neither implementable $2^{O(k)}$ value queries (by~\cite{DobzinskiV16}), nor implementable with $2^{O(k)}$ simultaneous communication (see \Cref{prop:SimultaneousLowerBound}). Such a mechanism would be fundamentally different than all prior MIR mechanisms, which can be implemented both simultaneously and with value queries.

We discuss these directions further in~\Cref{section:LogMGap}.

\medskip
\textbf{Beyond MIR mechanisms.}
The major open problem is to understand communication lower bounds that hold for all deterministic truthful mechanisms, and not just MIR mechanisms. There is significantly less progress in this direction -- only~\cite{DobzinskiRV22} for dominant strategy truthful mechanisms, and~\cite{AssadiKSW20} for two-player mechanisms.

\bibliographystyle{alpha}
\bibliography{MasterBib}

\appendix

\section{Proofs of Folklore Claims and Auxiliary Results} \label{sec:Folklore}

\subsection{MIR Mechanisms are as Powerful as Affine Maximizers}

Here, we give a proof that affine maximizers are not more powerful in our setting than MIR mechanisms.

\begin{claim}
    Let $\calA \subseteq N^M$ and let $\calF$ be a class of functions $f : 2^M \to \mathbb{R}$ closed under scaling. If every MIR mechanism for $\calA$ over $\calF$ is not $\alpha$-approximate and uses at least $\beta$ communication, then every affine maximizer for $\calA$ over $\calF$ is not $\alpha$-approximate and uses at least $\beta$ communication.
\end{claim}
\begin{proof}
Consider any affine maximizer with adjustment $v_0 : \calA \to \mathbb{R}$ and weights $c \in \mathbb{R}_+^n$. Since any MIR mechanism for $\calA$ produces a welfare-maximizing allocation over $\calA$, the affine maximizer cannot produce a better allocation, and hence cannot be $\alpha$-approximate.

Fix any $v_1, \dots, v_n \in \calF$, and define the function $v : N^M \to \mathbb{R}_+$ by $v(A) \coloneqq \sum_{i \in N} v_i(A(i))$. Let $c^+ \coloneqq \max_{A \in \calA} \abs{v_0(A)}$ and $c^- \coloneqq \min_{A, B \in \calA, v(A) \ne v(B)} \abs{v(A) - v(B)}$ (we assume that $v$ is not a constant function).

Now, consider the valuations $u_1, \dots, u_n$ defined $u_i \coloneqq (3c^+/c^-) v_i/c_i$. Note that $u_i \in \calF$ because $\calF$ is closed under scaling, and $c^+, c^-, c_i$ are non-negative. Then for the instance given by $u_1, \dots, u_n$, the affine maximizer finds an allocation $A^+$ which maximizes
\[
    v_0(A^+) + \sum_{i \in N} c_i u_i(A^+(i)) \quad = \quad v_0(A^+) + \frac{3c^+}{c^-} \sum_{i \in N} v_i(A^+(i)) \quad = \quad v_0(A^+) + \frac{3c^+}{c^-} v(A^+)\enspace .
\]

Suppose for contradiction that $A^+$ is not a welfare-maximizing allocation, and let $A^*$ be any welfare-maximizing allocation. Since the affine maximizer chose $A^+$, we have
\begin{align*}
    v_0(A^+) + \frac{3c^+}{c^-} v(A^+) \quad &\geq \quad v_0(A^*) + \frac{3c^+}{c^-} v(A^*) \\
    v_0(A^+) - v_0(A^*) \quad &\geq \quad \frac{3c^+}{c^-} (v(A^*) - v(A^+)) \enspace .
\end{align*}

However, by construction of $c^+$ and $c^-$, and because $v(A^*) > v(A^+)$ (or else $A^+$ would be a welfare-maximizing allocation), the inequality evaluates to $2c^+ \geq 3c^+$, which implies $c^+ = 0$, which implies $v_0 = 0$. This is a contradiction, because then the objective of the affine maximizer maximizes the scaled welfare, which is equivalent to maximizing the welfare. Therefore, $A^+$ must be a welfare-maximizing allocation for $v_1, \dots, v_n$.

Thus, we can use any affine maximizer for $\calA$ to run an MIR mechanism for $\calA$, so if every MIR mechanism uses at least $\beta$ communication, then every affine maximizer for $\calA$ must also use at least $\beta$ communication.
\end{proof}


\subsection{Computational Lower Bounds}

For the computational lower bounds in \Cref{table:Results}, we adapt the approaches of \cite{DanielySS15,DobzinskiV16} to give an optimal bound up to constants.

\begin{definition}[Almost-Single-Minded~\cite{DanielySS15}]
    Bidder $i$ is \emph{almost-single-minded} for $T \subseteq M$ if $v_i(S) = \mathbbm{1}_{T \subseteq S} + \abs{S}/m^3$.
\end{definition}

\begin{definition}[Menu~\cite{DobzinskiV16}]
    For a deterministic truthful mechanism and valuations $\vec{v}_{-i} \coloneqq (v_1, \dots, v_{i-1}, v_{i+1}, \dots, v_n)$, the \emph{menu} $\calM_{\vec{v}_{-i}} \subseteq 2^M$ is the collection of sets $S$ for which there exists $v_i$ such that when presented with valuations $\vec{v} = (\vec{v}_{-i}, v_i)$, the mechanism allocates $S$ to bidder $i$. The \emph{price} $p_{\vec{v}_{-i}} : \calM_{\vec{v}_{-i}} \to \mathbb{R}_+$ is a function from sets on the menu to prices charged by the mechanism; note that $p_{\vec{v}_{-i}}$ can always be extended to a monotone function over $2^M$, and so we will treat it as such. Observe that a truthful mechanism must allocate to bidder $i$ the set $S \in \calM_{\vec{v}_{-i}}$ which maximizes the utility $v_i(S) - p_{\vec{v}_{-i}}(S)$.
\end{definition}

\begin{definition}[Structured Submenu~\cite{DobzinskiV16}]
    A collection of sets $\calS \subseteq \calM_{\vec{v}_{-i}}$ is a \emph{structured submenu} if there exist $\ell, p$ such that for all $S \in \calS$,
    \begin{itemize}[topsep=0pt,itemsep=0pt]
        \item $\abs{S} = \ell$

        \item $p_{\vec{v}_{-i}}(S) \in [p - 1/m^5, p]$

        \item For all $T \in \calM_{\vec{v}_{-i}}$ such that $T \supset S$, $p_{\vec{v}_{-i}}(T) \geq p_{\vec{v}_{-i}}(S) + 1/m^3$
    \end{itemize}
\end{definition}
\vsthm

The construction in the following lemma is from \cite{DanielySS15}, but we need to do some extra work here to give a probabilistic result, rather than just an existence result.

\begin{lemma} \label{lemma:RichMenu}
    Let $k = \Omega(\log(m))$ and let $n = m/k$. If a deterministic truthful mechanism is $m/(4k)$-approximate, then there exists an efficient procedure to sample valuations $\vec{v}_{-i}$ such that $\abs{\calM_{\vec{v}_{-i}}}$ contains a structured submenu of size $2^{\Omega(k)}$ w.p. at least $1/n^2$
\end{lemma}
\begin{proof}
Let $T_1, \dots, T_n \subseteq M$ be a uniformly random list of disjoint sets of size $k/2$, and let bidder $i$ be almost-single-minded for $T_i$, and let $\calM'_{\vec{v}_{-i}}$ be the collection of sets $S$ for which there exists almost-single-minded $v_i$ such that the mechanism allocates $S$ to $v_i$.

Since $v_i(S) \leq 1 + 1/m^2$ for all $S \subseteq M$, $p_{\vec{v}_{-i}}(S) \leq 1 + 1/m^2$ for all $S \in \calM'_{\vec{v}_{-i}}$ or else bidder $i$ would have negative utility for $S$ and would never receive the set. Thus, we can partition $\calM'_{\vec{v}_{-i}}$ into $(m+1)(1+1/m^2)/(1/m^5) = \poly(m)$ submenus via choices of $\ell, p$ that satisfy the first two criteria of a structured submenu. To see that the third criterion is satisfied, observe that if $T \in \calM_{\vec{v}_{-i}}$ and $T \supset S$, then $v_i(T) - v_i(S) \geq 1/m^3$. Then if $p_{\vec{v}_{-i}}(T) - p_{\vec{v}_{-i}}(S) < 1/m^3$, almost-single-minded bidders $i$ would always prefer set $T$ to set $S$, so by truthfulness of the mechanism, $T$ would always be assigned to bidder $i$ over $S$, and $S \not\in \calM'_{\vec{v}_{-i}}$. Therefore, for any $T \in \calM_{\vec{v}_{-i}}$ such that $T \supset S$, we have $p_{\vec{v}_{-i}}(T) \geq p_{\vec{v}_{-i}}(S) + 1/m^3$.

Hence, each $\calM'_{\vec{v}_{-i}}$ contains a structured submenu of size $\abs{\calM'_{\vec{v}_{-i}}}/\poly(m)$. Since $k = \Omega(\log(m))$, it remains to show that we can efficiently sample $\vec{v}_{-i}$ such that $\abs{\calM'_{\vec{v}_{-i}}} = 2^{\Omega(k)}$.

Let $E_{i,S}$ be the event that $\abs{S} \leq m/4$ and $T_i \subseteq S$. Conditioned on $\vec{v}_{-i}$, $T_i$ is a uniformly random subset of $k/2$ items among $m/2 + k/2$ items, so
\begin{equation} \label{eq:EventBound}
    \Pr[E_{i,S} \given \vec{v}_{-i}] \quad \leq \quad \binom{m/2}{m/4-k/2}\bigg/\binom{m/2+k/2}{m/4} \quad \leq \quad 2^{-k/2} \enspace .
\end{equation}

Observe that if none of the events $E_{i,S}$ occur for $i \in N$, $S \in \calM'_{\vec{v}_{-i}}$, then no bidder $i$ is assigned a set $S \supseteq T_i$ unless $\abs{S} > m/4$, and so the welfare is at most $3 + 1/m^2$. However, we know that the instance has optimal welfare at least $n = m/k$, so the mechanism would not be $m/(4k)$ approximate. Therefore, for every $\vec{v}$, there exists some $i \in N$ and $S \in \calM'_{\vec{v}_{-i}}$ for which $E_{i,S}$ occurs. This implies that for some $i \in N$,
\[
    \E\bigg[\Pr\bigg[\bigcup_{S \in \calM'_{\vec{v}_{-i}}} E_{i,S} \:\bigg\vert\: \vec{v}_{-i}\bigg]\bigg] \quad \geq \quad \frac{1}{n} \enspace .
\]

By \eqref{eq:EventBound} and a union bound,
\[
    \E[\min\{1, \abs{\calM'_{\vec{v}_{-i}}} 2^{-k/2}\}] \quad \geq \quad \frac{1}{n} \enspace ,
\]
which implies
\[
    \E[\min\{2^{k/2}, \abs{\calM'_{\vec{v}_{-i}}}\}] \quad \geq \quad \frac{2^{k/2}}{n} \quad = \quad 2^{\Omega(k)} \enspace .
\]

As the random variable in the expectation is bounded, we have that
\[
    \Pr[\abs{\calM'_{\vec{v}_{-i}}} = 2^{\Omega(k)}] \quad \geq \quad \frac{1}{n} \enspace .
\]

Then to efficiently sample, we can generate $\vec{v}$ as described above and choose a random bidder $i$.
\end{proof}

Now, following the approach of~\cite{DobzinskiV16}, we use the properties of the large structured submenu to prove computational hardness. We will give a randomized reduction from a SAT instance with $\Theta(k)$ variables to running an $O(m/k)$-approximate deterministic truthful mechanism.

\begin{definition}
    For any matrix $X \in \{0, 1\}^{k \times m}$ and any set $S \subseteq M$, define $X(S) \in \{0, 1\}^k$ as the product entrywise modulo $2$ betwen $X$ and the indicator vector of $S$.
\end{definition}

\begin{definition}[The Bidding Language]
    Let $\vec{v}_{-i}$ be a set of valuations, let $\MECH$ be a $\Theta(m/k)$-approximate deterministic truthful mechanism which runs in time $\tau$ and has menu $\calM_{\vec{v}_{-i}}$ and prices $p_{\vec{v}_{-i}}$, let $\calS \subseteq \calM_{\vec{v}_{-i}}$ be a structured submenu with parameters $\ell, p$, let $X \in \{0, 1\}^{k \times m}$ be a matrix, and let $\phi$ be a formula with $k$ variables. Then we define $v^\phi_i \coloneqq v^\phi_i(\cdot \mutual \vec{v}_{-i}, \MECH, \ell, p, X)$ by
    \[
        v^\phi_i(S) \quad \coloneqq \quad \begin{dcases}
            3 & S \in \calS,\, \phi(X(S)) = 1 \\
            3 & \abs{S} > \ell,\, p_{\vec{v}_{-i}}(S) > p \\
            0 & o.w.
        \end{dcases} \enspace .
    \]
\end{definition}

\begin{proposition}[\cite{DobzinskiV16}, Lemma 3.20]
    The valuation $v^\varphi_i$ is monotone, can be described in $\poly(m, \tau)$ bits, and can be evaluated in $\poly(m, \tau)$ time.
\end{proposition}

\begin{lemma} \label{lemma:MaximizingSetIsSATSolution}
    If there exists $S \in \calS$ such that $\phi(X(S)) = 1$, then $\MECH$ allocates to bidder $i$ such a set on the instance $(\vec{v}_{-i}, v^\phi_i)$.
\end{lemma}
\begin{proof}
The utility of bidder $i$ for a set such that $S \in \calS$ and $\phi(X(S)) = 1$ is at least $3 - p > 0$, and so since the price of any other set that gives value $3$ is strictly larger, no other set can give bidder $i$ higher utility, and so a deterministic truthful mechanism must allocate a set such that $S \in \calS$ and $\phi(X(S)) = 1$ if it exists.
\end{proof}

\begin{proposition}[\cite{DobzinskiV16}, Claim 3.21] \label{prop:RandomProjectionUsuallyHits}
    Let $\phi$ be a satisfiable formula with $k$ variables, let $X \in \{0, 1\}^{k \times m}$ be uniformly random, and for any set $S \subseteq M$, denote $X(S)$ as the product modulo $2$ between $X$ and the indicator vector of $S$. Then if $\calS \subseteq 2^M$ and $\abs{\calS} \geq 2^{2k} + 1$, the probability that $\{X(S) : S \in \calS\}$ does not contain a satisfying assignment is at most $2^{-k}$.
\end{proposition}

\begin{lemma} \label{lemma:ComputationalLowerBound}
    Let $\phi$ be a formula with $k = \Omega(\log(m))$ variables, and let $\MECH$ be a $\Theta(m/k)$-approximate deterministic truthful mechanism which runs in time $\tau$. Then there exists a randomized algorithm which outputs NO w.p. $1$ when $\phi$ is unsatisfiable, outputs YES w.h.p. when $\phi$ is satisfiable, and runs in $\poly(m, \tau)$ time.
\end{lemma}
\begin{proof}
By \Cref{lemma:RichMenu}, we can efficiently sample valuations $v_{-i}$ such that w.p. at least $1/n^2$, $\MECH$ has a structured submenu $\calS$ with parameters $\ell, p$ of size $2^{\Omega(k)}$. While we do not know the parameters $\ell, p$ that yield a good structured submenu, there are only $\poly(m)$ possibilities for them, so we can guess them and still succeed w.p. $1/\poly(m)$.

Let $X \in \{0,1\}^{k \times m}$ be uniformly random. If $\phi$ is satisfiable and $\abs{\calS} = 2^{\Omega(k)}$, then by \Cref{prop:RandomProjectionUsuallyHits,lemma:MaximizingSetIsSATSolution}, $\MECH$ assigns bidder $i$ a set $S \subseteq M$ such that $\phi(X(S)) = 1$ w.p. at least $1-2^{-k}$ on the instance $(\vec{v}_{-i}, v^\phi_i)$, where $v^\phi_i \coloneqq v^\phi_i(\cdot, \vec{v}_{-i}, \MECH, \ell, p, X)$. By a union bound, running $\MECH$ on the instance $(\vec{v}_{-i}, v^\phi_i)$ allows us to recover a satisfying assignment of $\phi$ w.p. at least $1/\poly(m) - 2^{\Omega(k)} = 1/\poly(m)$.

Resampling $\vec{v}_{-i}, \ell, p, X$ and repeating the process $\poly(m)$ times causes a satisfying assignment to be recovered w.h.p. in $\poly(m, \tau)$ time. If a satisfying assignment is recovered, we output YES. If no satisfying assignment is recovered, we output NO. Thus, we answer correctly always when $\phi$ is unsatisfiable, and we answer correctly w.h.p. when $\phi$ is satisfiable.
\end{proof}

\begin{theorem} \label{thm:GeneralComputationalLowerBound}
    Let $k = \omega(\log(m))$. Assuming rETH, all $m/k$-approximate deterministic truthful mechanisms for succinctly representable valuations use $2^{\omega(k)}$ computation.
\end{theorem}
\begin{proof}
By \Cref{lemma:ComputationalLowerBound}, an $m/k$-approximate deterministic truthful mechanism for succinctly representable valuations running in time $\tau$ implies the existence of a randomized algorithm $\ALG$ running in time $\poly(m, \tau)$ which decides the satisfiability of any formula with $\Theta(k)$ variables. Assuming rETH, $\ALG$ runs in $2^{\Omega(k)}$ time. Since the polynomial dependence on $m$ only adds a $2^{\Theta(\log(m))}$ component to the running time of $\ALG$ and $k = \omega(\log(m))$, it must be the case that $\tau = 2^{\Omega(k)}$.
\end{proof}

\begin{theorem} \label{thm:SubmodularComputationalLowerBound}
    Let $k = \omega(\log(m))$. Assuming rETH, all $\sqrt{m/k}$-approximate deterministic truthful mechanisms for succinctly representable valuations use $2^{\omega(k)}$ computation.
\end{theorem}
\begin{proof}
\cite{DobzinskiV16} similarly shows how to efficiently generate instances $(\vec{v}_{-i}, v^\phi_i)$ for submodular valuations, analogous to our proof for general valuations. Then the same argument as \Cref{thm:GeneralComputationalLowerBound} suffices to give a tight bound assuming rETH.
\end{proof}


\subsection{Lower Bound for Simultaneous Protocols}

\begin{proposition} \label{prop:SimultaneousLowerBound}
    Let $\MECH$ be a $\sqrt{m/k}$-approximate MIR mechanism for submodular valuations. Then any simultaneous protocol implementing $\MECH$ must use $2^{\Omega(k)}$ communication.
\end{proposition}
\begin{proof}
Let $\calA$ be the allocation bank maximized over by $\MECH$ when $n = 3\sqrt{m/k}$. Then by \Cref{prop:SubadditiveManyAllocations}, $\abs{\calA} \geq 2^{\sqrt{mk}/3}$. This implies that there exists $\calA' \subseteq \calA$ and $i \in N$ such that
\begin{itemize}[topsep=0pt,itemsep=0pt]
    \item $\abs{\calA'} \geq (2^{\sqrt{mk}/3})^{1/n}/(m+1) = 2^{\Omega(k)}$

    \item For any distinct $A, B \in \calA'$, $\abs{A_i} = \abs{B_i}$ and $A_i \ne B_i$
\end{itemize}

Let $z \coloneqq \abs{\calA'}$, and index $\calA' \coloneqq \{A^{(1)}, \dots, A^{(z)}\}$. For any set $X \subseteq [z]$, let $v^{(X)}_i$ be the valuation of bidder $i$ if they are mild-desires (\Cref{def:MildDesires}) for $\{A^{(\ell)}_i : \ell \in X\}$. Observe that this is well-defined because $\abs{A_i} = \abs{B_i}$ for all $A, B \in \calA'$, and that distinct $X$ result in distinct valuations because $A_i \ne B_i$ for all $A, B \in \calA'$.

Suppose for contradiction that there exists a simultaneous protocol implementing $\MECH$ using $o(z)$ bits of communication. Then since there are only $2^{o(z)} < 2^z/(z+1)$ possible messages sent by bidder $i$, there exist distinct $X, Y \subseteq [z]$ such that $\abs{X} = \abs{Y}$, and bidder $i$ sends the same message when they have either the valuation $v^{(X)}_i$ or $v^{(Y)}_i$. Since $X \ne Y$ and $\abs{X} = \abs{Y}$, there exists $\ell_x \in X \setminus Y$, and $\ell_y \in Y \setminus X$. Then for each $j \ne i$, let bidder $j$ be mild-desires for $\{A^{(\ell_x)}_j, A^{(\ell_y)}_j\}$, and consider the instances
\begin{align*}
    & \vec{v}^{(X)} \coloneqq (v_1, \dots, v_{i-1}, v^{(X)}_i, v_{i+1}, \dots, v_n) \qquad \text{and} \\
    & \vec{v}^{(Y)} \coloneqq (v_1, \dots, v_{i-1}, v^{(Y)}_i, v_{i+1}, \dots, v_n) \enspace .
\end{align*}
Observe that the optimal welfare for $\vec{v}^{(X)}$ is $2m$, attained only by $A^{(\ell_x)}$, and that the optimal welfare for $\vec{v}^{(Y)}$ is $2m$, attained only by $A^{(\ell_y)}$. However, because each bidder sends the same message in both instances (bidder $i$ by construction, and the remaining bidders because their valuations across the instances are identical and the protocol is simultaneous), $\MECH$ cannot distinguish the two instances, and assigns the same allocation to both. Thus, $\MECH$ cannot be MIR, a contradiction. This implies that no simultaneous protocol can implement $\MECH$ using $o(z)$ bits of communication, and hence every simultaneous protocol implementing $\MECH$ uses $\Omega(z) = 2^{\Omega(k)}$ communication.
\end{proof}

\newpage
\section{Missing Tables and Proofs} \label{sec:MissingProofs}

\subsection{\Cref{table:Results} for General \texorpdfstring{$k$}{k}}

\bgroup
\newcommand\T{\rule{0pt}{3.1ex}}       
\newcommand\B{\rule[-1.7ex]{0pt}{0pt}} 
\def\arraystretch{1.5}
\begin{table}[ht]
    \centering\small
    \begin{tabular}{C{4.8cm}|C{4.8cm} C{5.5cm}}
Communication & General & Subadditive/XOS/Submodular \B \\ \hline
Prior Best Mechanisms & $O(m/\sqrt{k})$~\cite{HolzmanKMT04} & $O(\sqrt{m/k}\log(m/k))$~\cite{DobzinskiNS10} \T \\
Our MIR Mechanisms & \bm{$O(m/k)$}~(Thm.~\ref{thm:General}) & \bm{$O(\sqrt{m/k})$}~(Thm.~\ref{thm:Subadditive}) \\
Our MIR Lower Bounds & \bm{$\Omega(m/k)$}~(Thm.~\ref{thm:General}) & \bm{$\Omega(\sqrt{m/(k\log(m/k))})$}~(Thm.~\ref{thm:LowerBound}) \\
Prior MIR Lower Bounds & $\Omega(m/(k\log(m)))$~\cite{DanielySS15} & $\Omega((m/(k\log(m)))^{1/3})$~\cite{DanielySS15}
    \end{tabular}

    \bigskip
    
    \begin{tabular}{C{4.8cm}|C{4.8cm} C{5.5cm}}
Value Queries & General & Subadditive/XOS/Submodular \B \\ \hline
Prior Best Mechanisms & $O(m/\sqrt{k})$~\cite{HolzmanKMT04} & $O(\sqrt{m/k}\log(m/k))$~\cite{DobzinskiNS10} \T \\
Our MIR Mechanisms & \bm{$O(m/k)$}~(Thm.~\ref{thm:General}) & \bm{$O(\sqrt{m/k})$}~(Thm.~\ref{thm:Subadditive}) \\
Truthful Lower Bounds & $\downarrow$ & $\Omega(\sqrt{m/k})$~\cite{DobzinskiV16} \\
Algorithmic Lower Bounds & $\Omega(m/k)$~\cite{BlumrosenN05a} & $\Omega(\sqrt{m/k})$ for Subadditive/XOS~\cite{DobzinskiNS10}
    \end{tabular}

    \bigskip
    
    \begin{tabular}{C{4.8cm}|C{4.8cm} C{5.5cm}}
Computation (Succ.~Rep.) & General & Subadditive/XOS/Submodular \B \\ \hline
Prior Best Mechanisms & $O(m/\sqrt{k})$~\cite{HolzmanKMT04} & $O(\sqrt{m/k}\log(m/k))$~\cite{DobzinskiNS10} \T \\
Our MIR Mechanisms & \bm{$O(m/k)$}~(Thm.~\ref{thm:General}) & \bm{$O(\sqrt{m/k})$}~(Thm.~\ref{thm:Subadditive}) \\
Our Truthful Lower Bounds & \bm{$\Omega(m/k)$}~(Thm.~\ref{thm:GeneralComputationalLowerBound}) & \bm{$\Omega(\sqrt{m/k})$}~(Thm.~\ref{thm:SubmodularComputationalLowerBound}) \\
Prior Truthful Lower Bounds (only when $k = \poly(m)$) & $\Omega(m/(k\log(m)))$~\cite{DanielySS15} & $\Omega(\sqrt{m/k})$~\cite{DobzinskiV16}
    \end{tabular}

    \caption{Summary of our results (bolded) compared to prior work (unbolded). All models referenced above are for deterministic mechanisms. The listed expression is the approximation upper/lower bound when using $2^k$ of the respective resource.}
    
    \label{table:ResultsGeneral}
\end{table}
\egroup

\subsection{Computationally Efficient Mechanism for Subadditive Valuations}

\Balanced*

\BalancedBucketingsExist*
\begin{proof}
Fix $\vec{v}$ and let $P \in [t]^N$ be uniformly random. Observe that $\E[\abs{B^*_s(\vec{v}, P)}] \leq m/t$ for all $s \in [t]$, and that $\abs{B^*_s(\vec{v}, P)}$ is the sum of independent weighted Bernoulli random variables whose weights do not exceed $m/t$ (as we ignore all $i \not\in N_1(\vec{v})$). Thus, by a Chernoff bound, $\Pr[\abs{B^*_s(\vec{v})} = O(m/t)] = \Omega(1)$. By linearity of expectation, we have $\E[\OPT(\vec{v}, N_P) = \Theta(\OPT(\vec{v}, N_1)]$, so since $\OPT(\vec{v}, N_P) \leq \OPT(\vec{v}, N_1)$, we also have $\Pr[\OPT(\vec{v}, N_P) = \Theta(\OPT(\vec{v}, N_1)] = \Omega(1)$. If we take independent repetitions $P^{(1)}, \dots, P^{(y)} \in [t]^N$, the probability that none of the $P^{(x)}$ are balanced for $\vec{v}$ is at most $2^{-\poly(m)}$.

Now, let $\calN(\vec{v}) \subseteq 2^{N_1}$ be the collection of sets $N'$ such that $\OPT(\vec{v}, N') = \Theta(\OPT(\vec{v}, N_1))$. Since $\OPT(\vec{v}, \cdot)$ is a linear function on $\leq n$ variables, the number of distinct $\calN(\vec{v})$ is at most $2^{n^2}$~\cite{OrlikT13}. Note that the event $\{\OPT(\vec{v}, N_P) = \Theta(\OPT(\vec{v}, N_1))\}$ depends only on $\calN(\vec{v})$ and $A^*$. Thus, by a union bound over $2^{n^2} n^m = 2^{\poly(m)}$ possibilities, we have that $P^{(1)}, \dots, P^{(y)}$ is balanced w.p. $> 0$, which implies existence of balanced $P^{(1)}, \dots, P^{(y)}$.
\end{proof}

\EfficientBucketingMechanism*

\SubadditiveMechanismComputationallyEfficient*
\begin{proof}
Let $\calA$ define the mechanism. Let $N_0$ be the set of bidders who receive more than $2\sqrt{mk}$ items, and observe that $\MIR_\calA(\vec{v}) \geq 2\sqrt{k/m} \OPT(\vec{v}, N_0)$. Let $N_1 \coloneqq N \setminus N_0$, and observe that it is consistent with the definition of $N_1(\vec{v})$ from \Cref{def:Balanced}. Hence, we borrow the remaining notation from that definition.

Since $B^{(1)}, \dots, B^{(m)}$ is regular, let $f : \calP \to [m]$ be a function such that $B^{(\ell(P))}$ is regular for $B^*(\vec{v}, P) \coloneqq (B^*_1(\vec{v}, P), \dots, B^*_t(\vec{v}, P))$. We then define the valuation $v : \calP \to \mathbb{R}_+$ by
\[
    v(P) \quad \coloneqq \quad \sum_{s \in S_P(\vec{v})} \sum_{i \in P_s} v_i(B^{(f(P))}_s \cap A^*_i(\vec{v})) \enspace .
\]

\begin{lemma} \label{lemma:MechanismIsAtLeastVP}
    $\MIR_\calA(\vec{v}) \geq \max_{P \in \calP} v(P)$.
\end{lemma}
\begin{proof}
Since $B^{(f(P))}$ is regular for $B^*(\vec{v}, P)$, we have for $s \in S$ that
\[
    \sum_{i \in P_s \cap N_1} \abs{B^{(\ell(P))}_s \cap A^*_i(\vec{v})} \quad = \quad \abs{B^{(f(P))}_s \cap B^*_s(P)} \quad = \quad O(m/t^2) \quad = \quad O(k) \enspace .
\]
Thus, since $\calA^{(f(P))}_{s,P_s}$ shatters all subsets of $B^{(f(P))}_s$ of size $O(k)$, and since $\bigtimes_{s \in [t]} \calA^{(f(P))}_{s,P_s} \subseteq \calA$, we have $\MIR_\calA(\vec{v}) \geq v(P)$ for all $P \in \calP$.
\end{proof}

\begin{lemma} \label{lemma:PApproximatesP}
    \[
        \max_{P \in \calP} v(P) \quad \geq \quad 2\sqrt{\frac{k}{m}} \max_{x \in [y]} \OPT(\vec{v}, N_{P^{(x)}}) \enspace .
    \]
\end{lemma}
\begin{proof}
Consider any $P^{(x)}$, and observe that the shifts of $P^{(x)}$ are in $\calP$. Therefore, by the same subadditivity argument as in \eqref{eq:SubadditiveMainProperty}, for each $P^{(x)}$, there exists $P \in \calP$ such that
\[
    v(P) \quad \geq \quad 2\sqrt{\frac{k}{m}} \OPT(\vec{v}, N_{P^{(x)}}) \enspace .
\]
Taking maximums on each side completes the proof.
\end{proof}

Combining Lemmas~\ref{lemma:MechanismIsAtLeastVP},~\ref{lemma:PApproximatesP},~and~\ref{lemma:BalancedBucketingsExist}, we have
\[
    \MIR_\calA(\vec{v}) \quad \geq \quad \max_{P \in \calP} v(P) \quad \geq \quad 2\sqrt{\frac{k}{m}} \max_{x \in [y]} \OPT(\vec{v}, N_{P^{(x)}}) \quad = \quad \Theta\bigg(\sqrt{\frac{k}{m}} \OPT(\vec{v}, N_1)\bigg) \enspace .
\]
Scaling $k$ up by an appropriate constant gives a $2\sqrt{m/k}$ approximation of $\OPT(\vec{v}, N_1)$, which combined with the $2\sqrt{m/k}$ approximation of $\OPT(\vec{v}, N_0)$ yields the desired $\sqrt{m/k}$ approximation.

\medskip
\textbf{Communication and Computation.}
The proof that $\calA$ can be maximized over using $2^{O(k)}$ simultaneous value queries proceeds in the same way as in \Cref{thm:SubadditiveMechanismComputationallyInefficient}. For computation, observe that we only need to maximize over $mt\abs{\calP} = mzt^2 = \poly(m)$ different $\calA^{(\ell)}_{s,T}$ ($m$ different bucketings of the items, $\abs{\calP} = zt$ different bucketings of the bidders, $t$ different buckets), then find the maximum welfare over $zt = \poly(m)$ possibilities. Since $\calA^{(\ell)}_{s,T}$ can be optimized over in $2^{O(k)}$ computation by \Cref{thm:GeneralMechanism}, $\calA$ can also be optimized over in $2^{O(k)}$ computation.
\end{proof}

\subsection{Bound on the Number of Small Cuts for Lower Bound Constructions}

\CutCounting*
\begin{proof}
Consider a fixed $r$-way cut $C$ with at most $c\gamma_2(G)$ edges.

Suppose we contract a random edge of $G$, remove any resulting loops, and repeat until there are only $2c$ vertices remaining. Let $G_i$ denote the graph after the $i$\textsuperscript{th} iteration.

Observe that $\gamma_2(G_i) \geq \gamma_2(G)$ for all $i$, and thus the degree of each vertex in $G_i$ is at least $\gamma_2(G)$ (consider the cut with a single vertex on one side). Therefore, there are at least $(\abs{V}-i)\gamma_2(G)/2$ edges in $G_i$, and so the probability that we contract an edge in $C$ in the $i$\textsuperscript{th} iteration, conditioned on no edges in $C$ being contracted in prior iterations, is at most
\[
    \frac{c\gamma_2(G)}{(\abs{V}-i)\gamma_2(G)/2} \quad = \quad \frac{2c}{\abs{V}-i} \enspace .
\]

Therefore, the probability that no edges of $C$ are contracted at the end of the process is at least
\[
    \prod_{i=0}^{\abs{V}-2c-1} \bigg(1 - \frac{2c}{\abs{V}-i}\bigg) \quad \geq \quad \prod_{i=0}^{\abs{V}-2c-1} \frac{\abs{V}-i-2c}{\abs{V}-i} \quad \geq \quad \prod_{i=0}^{2c} \frac{1}{\abs{V}-i} \quad \geq \quad \abs{V}^{-2c} \enspace .
\]

Since there are at most $r^{2c}$ cuts in the final graph, we see that for \emph{any} $r$-way cut $C$ with at most $c\gamma_2(G)$ edges, the event $E_C$ that $C$ did not have any of its edges contracted in the final graph occurs w.p. at least $\abs{V}^{-2c}$ and occurs jointly with at most $r^{2c}$ other events $E_{C'}$. Thus, the number of $r$-way cuts $C$ with at most $c\gamma_2(G)$ edges must be bounded by $(\abs{V}r)^{2c} \leq \abs{V}^{4c}$.
\end{proof}

\section{The Logarithmic Gap} \label{section:LogMGap}

This section extensively details attempts and barriers to closing the logarithmic gap between the upper and lower bounds for submodular valuations.

We first go through a seemingly pointless exercise: an alternate $\sqrt{m/k}$-approximate MIR mechanism for XOS valuations that uses $2^{O(k)}$ communication (note that the result presented in \Cref{section:SubadditiveMechanism} additionally holds for subadditive mechanisms). However, this mechanism has a particular structure which demonstrates why closing the logarithmic gap is so difficult.


\subsection{An Alternate Mechanism for XOS Valuations} \label{section:AlternateXOSMechanism}

If $k = \Theta(m)$, then exactly maximizing the welfare only takes $2^{\Theta(m)}$ communication, so we focus on the regime where $k = o(m)$. Similarly, we have a trivial $2^{\Omega(\log(m))}$ communication lower bound, so we assume $k = \Omega(\log(m))$ (for a sufficiently large constant) as well.

First, let us define a generalization of the chunking mechanism from \Cref{section:GeneralMechanism}.

\begin{definition}[$r$-Chunking Mechanism]
    Let $B \in [t]^M$ partition $M$ into $t$ \emph{chunks}. The allocation bank $\calA_{B,r}$ contains every allocation where each bidder gets a set of the form $\bigcup_{s \in S} B_s$ for some $S \subseteq [t]$ such that $\abs{S} \leq r$, and the $r$-chunking mechanism for $B$ is MIR over $\calA_{B,r}$.
\end{definition}
\vsthm

The precise definition of the allocation bank we shall focus on is as follows:

\begin{definition} \label{def:AlternateXOSMechanismAllocationBank}
    WLOG, let $u \coloneqq \log\log(m/k)-1$, $\sqrt{mk}$, and $k/(\log(m/k))$ be integers. Let $U \coloneqq [u]$, and define
    \[
        \calX \quad \coloneqq \quad \bigg\{x \in \mathbb{Z}_+^u : \sum_{r \in U} x_r 2^r \leq 2\log(m/k)\bigg\} \enspace .
    \]
    For each $x \in \calX$ and $s \in [m^5]$, let $B^{(x,s)}$ be a randomized partition created as follows: first, create the uniformly random partition $C^{(x,s)} \in [\sqrt{mk}]^U$, and then merge chunks of $C^{(x,s)}$ until there are at least $x_r\sqrt{mk}/(2\log(m/k))$ chunks consisting of $2^r$ original chunks of $C^{(x,s)}$ for all $r \in U$.\footnote{This is possible because $\sum_{r \in U} x_r 2^r \leq 2\log(m/k)$.} Additionally, let $B^-$ be the partition of all singletons and let $B^+$ be the partition where all elements are in the same part. Let $\calA^{(x,s)}$ define the $k/(\log(m/k))$-bundling mechanism for $B^{(x,s)}$, $\calA^-$ define the $k/(\log(m/k))$-bundling mechanism for $\calA^-$, $\calA^+$ define the $1$-bundling mechanism for $\calA^+$, and $\calA \coloneqq \calA^- \cup \calA^+ \cup \bigcup_{x \in \calX} \bigcup_{s \in [m^5]} \calA^{(x,s)}$.
\end{definition}
\vsthm

\begin{remark} \label{remark:XOSMechanismCore}
    There is a lot to unpack here, and the proof that this is actually $O(\sqrt{m/k})$-approximate for XOS valuations gets somewhat messy. Importantly, the ``core'' of this mechanism are the $k/\log(m/k)$-chunking mechanisms, which is the main barrier to closing the $\log(m)$ gap. In particular, we know that the communication complexity of welfare maximization when every bidder gets at most $k/\log(m/k)$ items (chunks) is $2^{\Omega(k/\log(m/k))}$ and $2^{O(k)}$, but we don't know the exact answer. This problem and its implications are discussed in detail in later sections. We now return to the proof that the above mechanism is $O(\sqrt{m/k})$-approximate.
\end{remark}
\vsthm

\medskip
\textbf{Step 1: Partitioning the Bidders.}
Fix some optimal allocation $A^*$, and WLOG, suppose that $v_i(\{j\}) = 0$ for any $j \not\in A^*_i$. Note that we can only assume this because our mechanism can be implemented in $2^{\Theta(k)}$ value queries using a simultaneous protocol, and thus eventually asks each bidder their value for every possible set they could receive.

Similarly to our other mechanisms, we handle bidders based on the number of items they receive in some optimal allocation $A^*$. The welfare of bidders who receive $\Omega(\sqrt{mk})$ items can be $O(\sqrt{m/k})$-approximated because there are only $O(\sqrt{m/k})$ of them. This is handled by $\calA^+$.

The welfare of bidders who receive $O(\sqrt{mk}/\log(m/k))$ items can be $O(\sqrt{m/k})$ approximated by giving them each their favorite $k/\log(m/k)$ items that they receive in $A^*$. This is handled by $\calA^-$, which only takes $\binom{m}{k/\log(m/k)} = 2^{\Theta(k)}$ value queries.

Therefore, if we can $O(\sqrt{m/k})$ approximate the bidders who receive who receive $\Omega(\sqrt{mk}\log(m/k))$ and $O(\sqrt{mk})$ items, we can lose a constant factor of $3$ and guarantee an overall $O(\sqrt{m/k})$ approximation of the welfare. As such, we assume from here on that for all $i$, $\abs{A^*_i} = \Omega(\sqrt{mk}/\log(m/k))$ and $\abs{A^*_i} = O(\sqrt{mk})$ for sufficiently large and small constants.

For our approach, it turns out that we need to know roughly how many items each bidder gets in some optimal allocation. To this effect, call a bidder a $2^r$-bidder if $\abs{A^*_i}\log(m/k)/\sqrt{mk} \in [2^r, 2^{r+1})$ for $r \in U$. Since there aren't too many possible configurations (each bidder likes the same number of items up to an $O(\log(m/k))$ factor), we can just guess every possible configuration, and merge the allocation banks into one big allocation bank. This is essentially the purpose of $\calX$ in \Cref{def:AlternateXOSMechanismAllocationBank}. The following lemmas formalize this idea.

\begin{lemma} \label{lemma:ComprehensiveCalX}
    There are fewer than $x_r\sqrt{m/k}$ $2^r$-bidders for all $r$ for some $x \in \calX$, and $\abs{\calX} = o(m)$. 
\end{lemma}
\begin{proof}
For each $r$, let $z_r\sqrt{m/k}$ be the number of $2^r$-bidders. Because there are only $m$ items, and each $2^rr$-bidder likes at least $2^r\sqrt{mk}/(\log(m/k))$ items, it must hold that
\[
    \sum_{r \in U} z_r\sqrt{m/k} \frac{2^r\sqrt{mk}}{\log(m/k)} \quad \leq \quad m \qquad \implies \qquad \sum_{r \in U} z_r 2^r \quad \leq \quad \log(m/k) \enspace .
\]

Let $x \in \mathbb{Z}_+^u$ be such that $z_r \in (x_r-1, x_r]$ for all $r \in U$. Then we must have
\begin{align*}
    \sum_{r \in U} (x_r-1) 2^r \quad &\leq \quad \sum_{r \in U} z_r 2^r \\
    \sum_{r \in U} x_r 2^r - \sum_{r \in U} 2^r \quad &\leq \quad \log(m/k) \\
    \sum_{r \in U} x_r 2^r \quad &\leq \quad 2\log(m/k) \enspace ,
\end{align*}
so $x \in \calX$.

Finally, note that $x_r \leq 2\log(m/k)$ for all $r$, so the number of possible $x \in \calX$ is bounded by $(2\log(m/k))^u = 2^{O(\log\log(m/k)^2)} = o(m)$.
\end{proof}

Hence, from here on, fix $x \in \calX$ satisfying the above condition.

\medskip
\textbf{Step 2: Linear Valuations.}
Let $\ell = 2^r$. The next step is to show that there exists $A \in \calA^{(x,1)}$ such that for each $\ell$-bidder $i$, $\abs{A_i \cap A^*_i} = \Omega(\ell k/\log(m/k))$ w.p. $\Omega(1)$. Observe that this immediately implies an $O(\sqrt{m/k})$ approximation for $0$-$1$ additive valuations, and is the key idea behind why this mechanism works. The following technical lemma does the probabilistic heavy lifting.

\begin{lemma} \label{lemma:BallsAndBins}
    For $r \in U$, let $\calB_r \subseteq [\sqrt{mk}]$ be the chunks of $B^{(x,1)}$ which are formed from $\ell = 2^r$ chunks of $C^{(x,1)}$. For each $\ell$-bidder $i$, define $\calS_i \coloneqq \{B^{(x,1)}_j \in \calB_r : \abs{B^{(x,1)}_j \cap A^*_i} \geq \ell\}$. Then
    \[
         \Pr\bigg[\abs{\calS_i} \geq \frac{k}{2\log(m/k)}\bigg] \quad \geq \quad \frac{9}{10} \enspace .
    \]
\end{lemma}
\begin{proof}
Fix some $\ell$-bidder $i$. We will think of each item in $A^*_i$ as a ball, and each $B^{(x,1)}_j \in \calB_r$ as a bin. Since $B^{(x,1)}_j$ was formed by merging $\ell$ chunks of $C^{(x,1)}$, a uniformly random partition into $\sqrt{mk}$ chunks, each $B^{(x,1)}_j \in \calB_r$ contains a thrown item w.p. $\ell/\sqrt{mk}$. Additionally, there are at least $\abs{\calB_r} \geq \sqrt{mk}/(2\log(m/k))$ such bins. Since $\abs{A^*_i} \geq \ell\sqrt{mk}/\log(m/k)$, there are at least that many balls. Hence, we may instead consider this balls and bins problem.

Arbitrarily index the $t \coloneqq \abs{\calB_r}$ bins, and let $X_1, \dots, X_t$ be the number of balls in each bin. Let $Z_s$ be the indicator that $X_s \geq \ell$. Then for sufficiently large $m$,
\begin{align*}
    \E[Z_s] \quad &\geq \quad \binom{\ell\sqrt{mk}/\log(m/k)}{\ell} \bigg(\frac{\ell}{\sqrt{mk}}\bigg)^{\ell} \bigg(1 - \frac{\ell}{\sqrt{mk}}\bigg)^{\ell\sqrt{mk}/\log(m/k)-\ell} \\
    &\geq \quad \frac{1}{\ell!}\bigg(1 - \frac{\log(m/k)}{\sqrt{mk}}\bigg)^\ell\bigg(\frac{\ell\sqrt{mk}}{\log(m/k)}\bigg)^\ell \bigg(\frac{\ell}{\sqrt{mk}}\bigg)^\ell \frac{1}{e^{\ell^2/\log(m/k)}} \\
    &\geq \quad \frac{1}{3e^{\ell^2/\log(m/k)}}\bigg(\frac{e\sqrt{mk}}{\log(m/k)}\bigg)^\ell \bigg(\frac{\ell}{\sqrt{mk}}\bigg)^\ell \\
    &\geq \quad \frac{1}{3e^{\ell^2/\log(m/k)}} \bigg(\frac{e\ell}{\log(m/k)}\bigg)^\ell \enspace .
\end{align*}

Performing a change of variables $u = \ell/\log(m/k)$,
\[
    \E[Z_s] \quad \geq \quad \frac{(eu)^{u\log(m/k)}}{3e^{u^2\log(m/k)}} \quad = \quad \frac{1}{3}\bigg(\frac{(eu)^u}{e^{u^2}}\bigg)^{\log(m/k)} \enspace .
\]

Since $u \in [0, 1]$, it can be numerically verified that for sufficiently large $m$,
\[
    \E[Z_s] \quad \geq \quad \bigg(\frac{k}{m}\bigg)^{1/4} \enspace ,
\]
so
\[
    \E[\abs{\calS_i}] \quad \geq \quad t\bigg(\frac{k}{m}\bigg)^{1/4} \quad = \quad \frac{m^{1/4} k^{3/4}}{2\log(m/k)} \quad \geq \quad \omega\bigg(\frac{k}{\log(m/k)}\bigg)\enspace .
\]

Finally, observe that because $Z_1, \dots, Z_t$ are non-decreasing functions of $X_1, \dots, X_t$, which result from balls and bins, they are negatively associated and hence satisfy Chernoff bounds~\cite{Joag-DevP83,DubhashiR98}. Thus,
\[
    \Pr\bigg[\abs{\calS_i} \geq \frac{k}{2\log(m/k)}\bigg] \quad \geq \quad 1 - e^{-\omega(k/\log(m/k))} \quad \geq \quad 1 - o(1) \enspace . \qedhere
\]
\end{proof}

\begin{remark}
    The cleverness of the above lemma is that the expected number of balls in each bin is only $\ell^2/\log(m/k) \leq \ell$, but since the expectation is logarithmic (or smaller), the deviations are also logarithmic, so we still get enough good bins.
\end{remark}

\begin{lemma} \label{lemma:AllBiddersGetEnoughItems}
    There exists $A \in \calA^{(x,1)}$ such that for any $\ell$-bidder $i$, $\abs{A_i \cap A^*_i} = \Omega(\ell k/\log(m/k))$ w.p. at least $3/5$. Additionally, there exists a function $f : M \to M$ from a set to a subset of that set such that each item is distributed identically in $f(A_i)$.
\end{lemma}
\begin{proof}
Fix any $r$, let $\ell = 2^r$, and define $\calB_r$ and $\calS_i$ as before. Observe that by \Cref{lemma:BallsAndBins} and standard Chernoff bounds, a $1-o(1)$ fraction of the $\ell$-bidders $i$ satisfy $\abs{B^{(x,1)}_j \cap A^*_i} \geq \ell$ for at least $k/\log(m/k)$ chunks $B^{(x,1)}_j \in \calB_r$. Further, since each item goes into each $B^{(x,1)}_j \in \calB_r$ with equal probability, the distribution of $\calS_i$ is symmetric, and hence choosing a uniformly random subset $\calS'_i$ of size $k/(2\log(m/k))$ from $\calS_i$ yields a uniformly random subset of $\calB_r$ of size $k/(2\log(m/k))$.

Choose such $\calS'_i$ for the $1-o(1)$ fraction of $\ell$-bidders for which $\abs{\calS_i} \geq k/(2\log(m/k))$. Consider any chunk covered by some $\calS'_i$. Let there be $n_\ell \leq x_r\sqrt{m/k}$ $\ell$-bidders. Since $\abs{\calB_r} \geq n_\ell k/(2\log(m/k))$, the probability that a chunk is covered by some $\calS'_i$ is
\begin{align*}
    1 - \bigg(\binom{\abs{\calB_r}-1}{k/(2\log(m/k))}\bigg/\binom{\abs{\calB_r}}{k/(2\log(m/k))}\bigg)^{(1-o(1))n_\ell} \quad &= \quad 1 - \bigg(1 - \frac{k/(2\log(m/k))}{\abs{\calB_r}}\bigg)^{(1-o(1))n_\ell} \\
    &\geq \quad 1 - \exp\bigg(\frac{(1-o(1))n_\ell k}{2\log(m/k)\abs{\calB_r}}\bigg) \\
    &\geq \quad \bigg(1 - \frac{1}{e}\bigg)\frac{(1-o(1))n_\ell k}{2\log(m/k)\abs{\calB_r}} \enspace .
\end{align*}

Thus, the expected number of chunks covered is at least $(1-o(1))(1-1/e)n_\ell k/(2\log(m/k))$. If we assign each covered chunk to a random bidder covering that chunk, we have by symmetry and a union bound (with the event that $1-o(1)$ of the $\ell$-bidders satisfy $\abs{\calS_i} \geq k/(2\log(m/k))$) that the expected number of chunks assigned to each bidder is at least $(3/10)k/\log(m/k)$. Hence, w.p. at least $3/5$, each bidder is assigned $\Omega(k/\log(m/k))$ of the chunks where they have at least $\ell$ items because each bidder is assigned at most $k/(2\log(m/k))$ chunks. Thus, there exists some $A \in \calA^{(x,1)}$ such that $\abs{A_i \cap A^*_i} = \Omega(\ell k/\log(m/k))$ w.p. at least $3/5$.

Finally, observe that because $\calS'_i$ is uniformly random, and each item is dropped with equal probability (only dropped when a chunk is covered by multiple bidders, and the chunks covered by each bidder are independent), there is a function $f$ from $A_i$ to a subset of $A_i$ (precisely the $\Omega(k/\log(m/k))$ chunks with at least $\ell$ items each) such that each item is distributed identically in $f(A_i)$.
\end{proof}

Now, taking many repetitions and doing a clever union bound suffices to prove that the mechanism works for all linear valuations.

\begin{lemma} \label{lemma:LinearFunctionsApproximated}
    The MIR mechanism for $\calA^{(x,1)}, \dots, \calA^{(x,m^5)}$ is $O(\sqrt{m/k})$-approximate for linear valuations w.p. $> 0$.
\end{lemma}
\begin{proof}
Observe that \Cref{lemma:AllBiddersGetEnoughItems} is the same as saying that bidder $i$ receives each item in $A_i$ w.p. at least $\Omega(\sqrt{k/m})$. Therefore, by linearity of expectation, the MIR mechanism for $\calA^{(x,1)}$ is $O(\sqrt{m/k})$-approximate in expectation, which implies the MIR mechanism for $\calA^{(x,1)}$ is $O(\sqrt{m/k})$-approximate w.p. at least $\Omega(\sqrt{k/m})$. By taking $m^5$ independent repetitions, the MIR mechanism for $\bigcup_{s \in [m^5]} \calA^{(x,s)}$ is $O(\sqrt{m/k})$-approximate w.p. at least $1-\exp(-\Omega(m^4))$.

Define $v : (2^M)^n \to \mathbb{R}_+$ by
\[
    v(S_1, \dots, S_n) \quad \coloneqq \quad \sum_{i \in N} v_i(S_i) \enspace ,
\]
and observe that $v$ is a linear function of $nm$ variables. Note that the success of the mechanism only depends on what subsets of the $nm$ variables give $v(\cdot) = \Omega(\OPT\sqrt{k/m})$. Hence, it suffices to just union bound over the number of ways for an $nm$-dimensional hyperplane (corresponding to the linear function) to separate $2^{nm}$ points (the number of subsets of $nm$ variables), which is bounded by $2^{(nm)^2} = 2^{o(m^4)}$~\cite{OrlikT13}.

Thus, we have w.p. $> 0$ that the MIR mechanism for $\calA^{(x,1)}, \dots, \calA^{(x,m^5)}$ is $\Omega(\sqrt{m/k})$-approximate for all linear valuations.
\end{proof}

\medskip
\textbf{Wrapping Up.}
Suppose we have XOS valuations $v_1, \dots, v_n$, and let $u_i$ be an additive clause of $v_i$ such that $u_i(A^*_i) = v_i(A^*_i)$. Then observe that for any allocation $A$, it holds that
\[
    \sum_{i \in N} u_i(A_i) \quad \leq \quad \sum_{i \in N} v_i(A_i) \enspace ,
\]
and further,
\[
    \sum_{i \in N} u_i(A^*_i) \quad = \quad \sum_{i \in N} v_i(A^*_i) \enspace .
\]

Therefore, the $O(\sqrt{m/k})$ approximation guaranteed by \Cref{lemma:LinearFunctionsApproximated} automatically extends to XOS valuations. Finally, observe that $\calA^{(x,1)}$ can be optimized over in $2^{O(k)}$ value queries because each bidder can only get at most $\binom{\sqrt{mk}}{k/(2\log(m/k))} = 2^{O(k)}$ possible sets, and taking $m^5\abs{\calX} = 2^{O(k)}$ repetitions does not affect the asymptotic communication. Combining this with $\calA^-$ and $\calA^+$ gives a mechanism which is $O(\sqrt{m/k})$-approximate for XOS valuations.

\begin{remark}
    It turns out that when certain properties are satisfied (and they are satisfied here), the median of subadditive functions is only a constant factor off from their expectation, which is only an $\Omega(\sqrt{k/m})$ factor from optimal. Therefore, we can make a similar argument as with linear functions that in expectation (and hence with not-too-small-probability), the mechanism is $O(\sqrt{m/k})$-approximate for subadditive valuations as well.

    The trouble with extending the result all the way to subadditive valuations actually comes from the union bound step; our reduction from XOS to linear valuations relies on being able to find, for any set $S$, a linear function $u_i$ such that $u_i \leq v_i$ but also $u_i(S) = \Theta(v_i(S))$. On the other hand, there exist subadditive functions for which no such linear functions exist, so it does not suffice to union bound over the class of linear functions.
\end{remark}


\subsection{Collection Disjointness}

As alluded to in \Cref{remark:XOSMechanismCore}, what makes the alternate XOS mechanism interesting is that it reduces to maximizing the welfare in a combinatorial auction with $n$ bidders and $m = nk/\log(n)$ items, where each bidder gets at most $k/\log(n)$ items. Doing this in $2^{O(k)}$ communication is easy (ask everyone for every set of size $k/\log(n)$), and a $2^{\Omega(k/\log(n))}$ lower bound is similarly simple (the allocation bank for $2$ bidders and $2k/\log(n)$ items already requires $2^{\Omega(k/\log(n))}$ communication to maximize over). Improving upon this upper bound directly improves the upper bound on $\MIRsubmod(m, k)$, and improving upon this lower bound is a necessary step for an improved lower bound on all MIR mechanism.

The heart of the problem can be distilled to the following question (where we have renotated $k/\log(n)$ to $k$ for simplicity):

\begin{definition}[\textsc{CollectionDisjointness}]
    Let there be items $M = [nk]$ and players $N = [n]$ who each hold a collection of sets $\calS_1, \dots, \calS_n \subseteq 2^M$, where each set $S \in \calS_i$ is of size $k$. Determine whether there exist $S_1 \in \calS_1, \dots, S_n \in \calS_n$ such that $S_1, \dots, S_n$ are disjoint. In other words, does there exist an allocation $S_1, \dots, S_n$ such that $S_i \in \calS_i$ for all $i \in N$? If so, we call the instance a YES instance, and if not, we call it a NO instance.
\end{definition}
\vsthm

The main question is whether the communication complexity of deciding \textsc{CollectionDisjointness} is $2^{\Theta(k)}$ or $2^{\Theta(k\log(n))}$, or something in between. Currently, we suspect the correct answer is $2^{\Theta(k\log(n))}$ due to the approach outlined in the following section.


\subsection{Lower Bound Approaches}

The natural first attempt would be try and replicate the lower bound technique of \Cref{section:LowerBounds}, i.e., randomly sample $2^{\Theta(k\log(n))}$ allocations and reduce from \textsc{SetDisjointness}. Unfortunately, this fails because by the time we sample even just $2^{\Theta(k)}$ allocations, the property in \Cref{claim:Structure} w.h.p. does not occur. Instead, we will use an approach inspired by the frameworks developed in~\cite{Dobzinski16b}.

\begin{definition}[Menu]
    For an instance $\calS_1, \dots, \calS_n$ of \textsc{CollectionDisjointness}, the \emph{menu} $\calM \coloneqq \calM(\calS_3, \dots, \calS_n)$ is defined
    \[
        \calM \quad \coloneqq \quad \{S \subseteq M : \abs{S} = 2k,\, \exists S_3 \in \calS_3, \dots, S_n \in \calS_n \text{ s.t. $S, S_3, \dots, S_n$ are disjoint}\} \enspace .
    \]
\end{definition}

\begin{lemma} \label{lemma:SufficientForLowerBound}
    If there exists an instance $\calS_1, \dots, \calS_n$ of \textsc{CollectionDisjointness} such that for all distinct $S, T \in \calM$, $\abs{S \cap T} < k$, then the communication complexity of \textsc{CollectionDisjointness} is $\Omega(\abs{\calM})$.
\end{lemma}
\begin{proof}
Arbitrarily index the sets of $\calM$ by $S^{(1)}, \dots, S^{(\abs{\calM})}$, and for each $S^{(\ell)}$, let $S^{(\ell)}_1, S^{(\ell)}_2$ be an arbitrary partition of $S^{(\ell)}$ into two sets of size $k$.

Now, consider a \textsc{SetDisjointness} instance with two sets $X_1, X_2 \subseteq [\abs{\calM}]$. We then define $\calS'_i$ for $i \in \{1, 2\}$ by $\calS_i \coloneqq \{S^{(\ell)}_i : \ell \in X_i\}$. We claim that $X_1 \cap X_2 \ne \emptyset$ if and only if the \textsc{CollectionDisjointness} instance $\vec{\calS} \coloneqq (\calS'_1, \calS'_2, \calS_3, \dots, \calS_n)$ is a YES instance.

Observe that by definition of $\calM$, $\vec{\calS}$ to be a YES instance if and only if there exists $S_1 \in \calS'_1$ and $S_2 \in \calS'_2$ such that $(S_1 \cup S_2) \in \calM$. Therefore, we simply need to show that there exists such $S_1, S_2$ if and only if $X_1 \cap X_2 \ne \emptyset$.

The reverse direction is trivial; simply take $\ell \in X_1 \cap X_2$ and consider the sets $S^{(\ell)}_1, S^{(\ell)}_2$. For the forward direction, suppose for contradiction that there exists $\ell_1, \ell_2, \ell$ such that $S^{(\ell_1)}_1 \cup S^{(\ell_2)}_2 = S^{(\ell)}$, but $X_1 \cap X_2 = \emptyset$, i.e., we cannot have $\ell_1 = \ell_2$. Then WLOG, let $\ell_1 \ne \ell$. Since $\abs{S^{(\ell_1)} \cap S^{(\ell)}} < k$ by assumption, we have $S^{(\ell_1)}_1 \cap S^{(\ell)} < k$. Therefore, we must have $S^{(\ell_2)}_2 \cap S^{(\ell)} > k$, which is impossible because $\abs{S^{(\ell_2)}_2} = k$.

Thus, any \textsc{SetDisjointness} instance on $[\abs{\calM}]$ can be reduced to a \textsc{CollectionDisjointness} instance. Since the communication complexity of \textsc{SetDisjointness} is $\Omega(\abs{\calM})$, the communication complexity of \textsc{CollectionDisjointness} must be $\Omega(\abs{\calM})$ as well.
\end{proof}

The difficulty then is proving that there exists an instance with the desired property for $\abs{\calM} = 2^{\Theta(k\log(n))}$. While only relatively weak assumptions are needed (e.g., a weak bound on the correlations between sets in $\calM$), the main problem is working with a distribution on $\calM$. For example, the most straightforward attempt would be to sample each $\calS_i$ by independently including each set of size $k$ with some probability. Unfortunately, this runs into elementary problems.

\begin{lemma} \label{lemma:AlmostNothingInMenu}
    Let $n$ be sufficiently large and $k = o(n)$. Let $\calS_i$ be sampled by independently including each $S \in \calS_i$ independently w.p. $1/(n-2)^k$, and let $\calS_3, \dots, \calS_n$ be independent as well. Then w.h.p., $\abs{\calM} = 0$.
\end{lemma}
\begin{proof}
Fix any set $S \subseteq M$ of size $2k$, and observe that there are $(nk)!/(k!^{n-2}(2k)!)$ tuples of sets $(S_3, \dots, S_n)$ that would cause $S \in \calM$. Let $Z_S$ be the number of these tuples that exist. Then
\begin{align*}
    \E[Z_S] \quad &= \quad \frac{((n-2)k)!}{k!^{n-2}} \frac{1}{(n-2)^{(n-2)k}} \\
    &\leq \quad \frac{((n-2)k)^{(n-2)k}}{k^{(n-2)k}\sqrt{2\pi k}^{n-2}} \frac{1}{(n-2)^{(n-2)k}} \\
    &\leq \quad \frac{1}{k^{n/2-1}} \\
    &\ll \quad \binom{nk}{2k}^{-1} \enspace .
\end{align*}
Markov's Inequality and a union bound completes the proof.
\end{proof}

\begin{lemma} \label{lemma:AlmostEverythingInMenu}
    Let $n$ be sufficiently large and $k = \Omega(\log(n))$. Let $\calS_i$ be sampled by independently including each $S \in \calS_i$ independently w.p. $e^k/(n-2)^k$, and let $\calS_3, \dots, \calS_n$ be independent as well. Then w.h.p., $\abs{\calM} = (1-o(1))\binom{nk}{2k}$.
\end{lemma}
\begin{proof}
Using the same notation as in \Cref{lemma:AlmostNothingInMenu}, we have
\[
    \E[Z_S] \quad = \quad \frac{((n-2)k)!}{k!^{n-2}} \frac{e^{(n-2)k}}{(n-2)^{(n-2)k}} \enspace ,
\]
and
\[
    \Var[Z_S] \quad \leq \quad \frac{((n-2)k)!}{k!^{n-2}}\frac{e^{(n-2)k}}{(n-2)^{(n-2)k}} \sum_{i=0}^{n-3} \binom{n}{i} \frac{(ik)!}{k!^i}\frac{e^{ik}}{(n-2)^{ik}} \enspace .
\]

\begin{claim} \label{claim:IrrelevantCombinatorialLemma}
    $f_{n,k} : \{0, \dots, n-3\} \to \mathbb{R}$ is maximized at $n-3$, where we define
    \[
        f_{n,k}(i) \quad \coloneqq \quad \binom{n}{i} \frac{(ik)!}{k!^i} \frac{e^{ik}}{(n-2)^{ik}} \enspace .
    \]
\end{claim}
\begin{proof}
For $i \in [n-3]$,
\[
    \frac{f_{n,k}(i)}{f_{n,k}(i-1)} \quad = \quad \frac{e^k}{(n-1)^k}\frac{(ik)!}{k!((i-1)k)!}\binom{n}{i}\bigg/\binom{n}{i-1} \quad = \quad \frac{e^k}{(n-1)^k} \binom{ik}{k} \frac{n-i+1}{i} \enspace .
\]
We see that $f_{n,k}$ is increasing when $i \geq n/2$ because
\[
    \frac{e^k}{(n-1)^k} \binom{ik}{k} \frac{n-i+1}{i} \quad \geq \quad \frac{e^k}{(n-1)^k} \bigg(\frac{n}{2}\bigg)^k \frac{1}{n} \quad \geq \quad 1
\]
For sufficiently large $n$ and $k = \Omega(\log(n))$. Further, for $i \leq n/2$ and sufficiently large $n$ and $k$,
\[
    f_{n,k}(i) \quad \leq \quad 2^n \frac{(ik)^{ik}}{k^{ik}} \frac{e^{ik}}{(n-1)^{ik}} \quad \leq 2^n e^{ik} \quad \leq \quad e^{3nk/4} \enspace ,
\]
while
\[
    f_{n,k}(n-3) \quad \geq \quad \frac{((n-3)k)^{(n-3)k}}{k^{(n-3)k} k^{n-3}} \frac{e^{(n-3)k}}{(n-2)^{(n-3)k}} \quad \geq \quad \bigg(\frac{n-3}{n-2}\bigg)^{(n-3)k} \frac{e^{(n-3)k}}{k^{n-3}} \quad \geq \quad e^{3nk/4} \enspace . \qedhere
\]
\end{proof}

By \Cref{claim:IrrelevantCombinatorialLemma}, we have
\begin{align*}
    \frac{\Var[Z_S]}{\E[Z_S]^2} \quad &\leq \quad \frac{k!^{n-2}}{((n-2)k)!}\frac{(n-2)^{(n-2)k}}{e^{(n-2)k}} \sum_{i=0}^{n-3} \binom{n}{i} \frac{(ik)!}{k!^i}\frac{e^{ik}}{(n-2)^{ik}} \\
    &\leq \quad n^4\frac{k!^{n-2}}{((n-2)k)!}\frac{(n-2)^{(n-2)k}}{e^{(n-2)k}} \frac{((n-3)k)!}{k!^{n-3}}\frac{e^{(n-3)k}}{(n-2)^{(n-3)k}} \\
    &\leq \quad n^5\frac{k!}{e^k} \frac{((n-3)k)!}{((n-2)k)!} \\
    &\leq \quad n^5 k \frac{k^k}{e^{2k}} \frac{1}{(nk)^k} \bigg(\frac{n}{n-3}\bigg)^k \\
    &\leq \quad (en)^{-k} \enspace .
\end{align*}
Chebyshev's Inequality yields $\E[\abs{\calM}] \geq (1-o(1))\binom{nk}{2k}$, and $\abs{\calM} \leq \binom{nk}{2k}$ yields the result w.h.p.
\end{proof}

In \Cref{lemma:AlmostNothingInMenu}, $\E[\abs{\calS_i}] \approx e^k$, whereas in \Cref{lemma:AlmostEverythingInMenu}, $\E[\abs{\calS_i}] \approx e^{2k}$. This demonstrates that there is a sharp transition between when nothing is in the menu, to when almost everything is in the menu, so even finding $\calM$ such that $\abs{\calM} = 2^{\Theta(k\log(n))}$ for sufficiently small constant is difficult, let alone analyzing any properties of $\calM$. Nevertheless, \Cref{lemma:SufficientForLowerBound} gives a concrete direction for an improved lower bound.

\end{document}